\documentclass[a4paper,UKenglish,cleveref, autoref, thm-restate]{lipics-v2021}

\usepackage{microtype}

\usepackage[toc,page]{appendix}

\usepackage{booktabs}

\usepackage[utf8]{inputenc}

\usepackage{xcolor}
\definecolor{lgray}{gray}{0.6}
\definecolor{dgray}{gray}{0.35}
\usepackage{tikz}
\usetikzlibrary{arrows,shapes,trees, arrows.meta}

\makeatletter
\@fleqnfalse
\makeatother

\newtheorem{rules}{}

\newtheorem*{theorem*}{Theorem}

\title{Stabilization Bounds for Influence Propagation from a Random Initial State}

\author{Pál Andras Papp}{ETH Zürich, Switzerland}{apapp@ethz.ch}{}{}
\author{Roger Wattenhofer}{ETH Zürich, Switzerland}{wattenhofer@ethz.ch}{}{}

\authorrunning{P.A. Papp and R. Wattenhofer}

\Copyright{Pál Andras Papp and Roger Wattenhofer}



\EventEditors{Filippo Bonchi and Simon J. Puglisi}
\EventNoEds{2}
\EventLongTitle{46th International Symposium on Mathematical Foundations of Computer Science (MFCS 2021)}
\EventShortTitle{MFCS 2021}
\EventAcronym{MFCS}
\EventYear{2021}
\EventDate{August 23--27, 2021}
\EventLocation{Tallinn, Estonia}
\EventLogo{}
\SeriesVolume{202}
\ArticleNo{48}
\nolinenumbers
\hideLIPIcs 

\begin{document}

\maketitle

\vspace{12pt}

\begin{abstract}
We study the stabilization time of two common types of influence propagation. In majority processes, nodes in a graph want to switch to the most frequent state in their neighborhood, while in minority processes, nodes want to switch to the least frequent state in their neighborhood. We consider the sequential model of these processes, and assume that every node starts out from a uniform random state.

We first show that if nodes change their state for any small improvement in the process, then stabilization can last for up to $\Theta(n^2)$ steps in both cases. Furthermore, we also study the proportional switching case, when nodes only decide to change their state if they are in conflict with a $\frac{1+\lambda}{2}$ fraction of their neighbors, for some parameter $\lambda \in (0,1)$. In this case, we show that if $\lambda < \frac{1}{3}$, then there is a construction where stabilization can indeed last for $\Omega(n^{1+c})$ steps for some constant $c>0$.
On the other hand, if $\lambda > \frac{1}{2}$, we prove that the stabilization time of the processes is upper-bounded by $O(n \cdot \log{n})$.
 \end{abstract}

\section{Introduction}

Dynamically changing colorings in a graph can be used to model various situations when entities of a network are in a specific state, and they occasionally decide to change their state based on the states of their neighbors. Such colorings are essentially a form of distributed automata, where the nodes can represent anything from brain cells to rival companies; as such, the study of these processes has applications in almost every branch of science.

One prominent example of such colorings is a \textit{majority process}, where each node wants to switch to the color that is most frequent in its neighborhood. These processes are used to model a wide range of phenomena in social sciences, e.g. the spreading of political opinions in social networks, or the adoption of different social media platforms \cite{MajApplic4, MajApplic2, MajApplic5}.

Another example is the dual setting of a \textit{minority process}, where each node wants to switch to the least frequent color among its neighbors. Minority processes can model settings where nodes would prefer to differentiate from each other, e.g. frequency selection in wireless networks, or selecting a production strategy in a market economy \cite{MinApplic4, KPRanticoor, MinApplic3}.

In our paper, we analyze the stabilization time of majority and minority processes, i.e. the number of steps until no node wants to change its color anymore. We study the processes in the \textit{sequential} (or asynchronous) model, where in every step, exactly one node switches its color. As such, stabilization time in the sequential model describes the total number of switches before the process terminates.

Compared to a synchronous setting, the sequential model has the advantage that neighbors are never switching at the exact same time; this prevents the process from ending up in an infinitely repeating periodic pattern. This property is indeed a reasonable assumption in many application areas, including the examples mentioned above: you are highly unlikely to e.g. switch your wireless frequency at the exact same time as your neighbors, or change your political opinion at the exact same time as your friends.

We study the maximal stabilization time of the processes in general graphs, assuming that the initial coloring of nodes is chosen uniformly at random. This setting may be relevant for a worst-case analysis in applications where the only thing we can influence is the initial coloring. For example, a wireless service provider might have no control over the topology of the network or the times when clients decide to switch their frequency, but it could easily ensure that its devices are initialized with a randomly chosen frequency.

An important parameter of the model is the switching rule, i.e. the threshold at which a node decides to switch to the opposite color. Two very natural rules are (i) basic switching, when nodes decide to switch for any small improvement, and (ii) proportional switching, when we have a real parameter $\lambda \in (0,1)$, and nodes only change their color if they are motivated to switch by a $\frac{1+\lambda}{2}$ fraction of their neighborhood. 

In our paper, we study the stabilization time for both basic and proportional switching. As a warm-up (in Section \ref{sec:basic}), we first show that in case of basic switching, both minority and majority processes can take $\Omega(n^2)$ steps to stabilize with high probability, matching a naive upper bound of $O(n^2)$. This follows from an extension of the lower-bound construction in \cite{minority} to the random-initialized case.

Our main contributions (Sections \ref{sec:lower} and \ref{sec:upper}) are stabilization bounds in case of proportional switching:

\begin{itemize}

\setlength\topsep{-12pt}
\setlength\itemsep{6pt}

\item for proportional switching with $\lambda < \frac{1}{3}$, we present a construction that w.h.p. exhibits a superlinear stabilization time of $\Omega(n^{1+c})$ for a constant $c>0$ that depends on $\lambda$.

\item for proportional switching with $\lambda > \frac{1}{2}$, we show that w.h.p. the process always stabilizes in $O(n \cdot \log{n})$ steps, essentially matching a straightforward lower bound of $\Omega(n)$.

\vspace{10pt}

\end{itemize}

\section{Related work}

Majority and minority processes have been extensively studied from numerous different perspectives since the early 1980s \cite{Goles, Goles2}. Most of the results focus on the simplest case of two colors, since this already captures the interesting properties of the process, and a generalization to more colors is often straightforward.

Many different variants of these processes have been inspired by application areas ranging from particle physics to social science, as in case of e.g. Ising systems or the voter model \cite{ising, votermodel}. In particular, there is extensive literature on more sophisticated process definitions that aim to provide a more realistic model for a specific application, such as social opinion dynamics or virus infection spreading \cite{SocialGen1, SocialGen2, Infection, useless}.

In case of majority processes, there is a particular interest in analyzing how a small set of nodes can influence the final state \cite{propDynamos1, propDynamos2, MajOther1, Schoenebeck, extra1}. For both processes, there are also numerous works on the analysis of stable states \cite{hedetniemi, SGPsurvey, KPRanticoor, votingtime, approx0}. However, in contrast to our work, most of these earlier results assume a synchronous setting, and only study the process on specific graph topologies, e.g. cliques, grids or Erdős-Rényi random graphs.

There is a recent line of work on stabilization time in general graphs; however, these results assume a worst-case initial coloring. For basic switching, the work of \cite{majority} shows that in the sequential adversarial and synchronous models, stabilization can last for $\widetilde{\Omega}(n^2)$ steps, matching a straightforward upper bound of $O(n^2)$. A similar lower bound is known for minority processes \cite{minority}. On the other hand, the two processes exhibit very different behavior in a benevolent sequential model: majority processes always stabilize in $O(n)$ time, while minority processes can last for quadratically many steps \cite{majority, minority}.

On the other hand, if we consider general graphs with proportional switching, then the sequential processes are known to exhibit a worst-case runtime between quadratic and linear, depending on the parameter $\lambda$ of the switching rule \cite{prop}. Stabilization time in this case is characterized by a non-elementary function $f(\lambda)$ that monotonically and continuously decreases from $1$ to $0$ on the interval $[0,1]$. The results of \cite{prop} show that for any $\varepsilon>0$, stabilization time is upper-bounded by $O(n^{1+f(\lambda)+\varepsilon})$, and the process can indeed last for $\Omega(n^{1+f(\lambda)-\varepsilon})$ steps. Our results are an interesting contrast to this, showing that if we randomize the initial state, then the process can only take $\Omega(n^{1+c})$ steps for smaller $\lambda$ values.

For general weighted graphs and a worst-case initial coloring, an exponential lower bound has also been shown for both majority \cite{majorityW} and minority \cite{minorityW} processes.

There are also various works that assume a randomized initial coloring, but these results focus on special classes of graphs. For majority processes, stabilization time from a randomized initial state has been analyzed in Erdős-Rényi random graphs, grids, tori and expanders \cite{majOther2, Ahad2018, extra2, extra3}. For minority processes, the works of \cite{CA1, CA2, CA3} study stabilization in cliques, cycles, trees and tori. As such, to our knowledge, stabilization time from a randomized initial coloring has not yet been studied in general graphs.

\section{Model definition and tools}

\subsection{Preliminaries}

We study the processes on simple, unweighted, undirected graphs $G(V,E)$ with node set $V$ and edge set $E$. We denote the nodes of the graph by $u$ or $v$, and the number of nodes in the graph by $n$. For a specific node $v$, we denote the neighborhood of $v$ by $N(v)$, and the degree of $v$ by $d_v=|N(v)|$. For ease of presentation, we usually define the size of our graph constructions in terms of an (almost) linear parameter $m$, and in the end, we select a value of $m$ that ensures $m \in \widetilde{\Theta}(n)$.

As common in this area, we focus on the case of two colors. That is, we say that a \textit{coloring} of the graph is a function $\gamma: V \rightarrow \{\text{black, white}\}$. For a specific coloring $\gamma$, we define $N_s(v)=\{ u \in N(v) \, | \, \gamma(v) = \gamma(u)\}$ as the neighbors of $v$ with the same color, and $N_o(v)=\{ u \in N(v) \, | \, \gamma(v) \neq \gamma(u)\}$ as the neighbors of $v$ with the opposite color.

We use the concept of \textit{conflicts} to define both majority and minority processes in a general form. We say that there is a \textit{conflict} on the edge $(u,v)$ if this edge motivates $v$ to change its color; more formally, if $u \in N_o(v)$ in case of a majority process, and if $u \in N_s(v)$ in case of a minority process. We use $N_c(v)$ to denote the conflicting neighbors of $v$ under $\gamma$, i.e. $N_c(v)=N_o(v)$ for majority and $N_c(v)=N_s(v)$ for minority.

Given a specific coloring $\gamma$, we say that node $v$ is \textit{switchable} if $|N_c(v)|$ is larger than a specific threshold, which is defined by the so-called \textit{switching rule} (discussed in detail in the next subsection). If $v$ is switchable, then it can change its color to the opposite color (i.e. it can \textit{switch}). We also use the word \textit{balance} to refer to the metric $\frac{|N_c(v)|}{d_v}$ in general, which indicates how close node $v$ is to being switchable.

A \textit{majority/minority process} is a sequence of colorings of the graph $G$ (known as \textit{states}). Every state is obtained from the previous state by switching a switchable node in the previous state. We assume that exactly one node switches in each step, which is often known as the \textit{sequential} or asynchronous model of the process. In our paper, we also assume that the initial state of the process is a \textit{uniform random coloring}, i.e. each node is white with probability $\frac{1}{2}$ and black with probability $\frac{1}{2}$, independently from other nodes.

We say that a state of the process is \textit{stable} if there are no more switchable nodes in the graph. The number of steps in the process (from the initial state until a stable state is reached) is known as the \textit{stabilization time} of the process.

We study the processes in general graphs, and we are interested in the longest possible stabilization time of a process, i.e. if in each step, the next node to switch among the switchable nodes is selected by an adversary who maximizes stabilization time. In other words, we study the worst-case stabilization of a graph on $n$ nodes under the worst possible ordering of switches.

We also use basic tools from probability theory, such as the union bound and the Chernoff bound, and the concept of an event happening \textit{with high probability} (\textit{w.h.p.}). For completeness, a brief summary of these techniques is provided in Appendix \ref{App:Z}.

\subsection{Switching rules}

Another important parameter of the processes is the condition that allows nodes to switch their color. There are two natural candidates for such a switching rule:

\vspace{2pt}

\begin{rules}
\textbf{$\!\!$\emph{Basic switching:}}   node $v$ is switchable if $|N_c(v)| \, > \, \frac{1}{2} \cdot d_v$.
\end{rules}

\begin{rules}
\textbf{$\!\!$\emph{Proportional switching:}}   node $v$ is switchable if $|N_c(v)| \, \geq \, \frac{1+\lambda}{2} \cdot d_v$.
\end{rules}

\vspace{2pt}

Note that both rules ensure that the overall number of conflicts in the graph strictly decreases in each switching step. Since there are at most $|E|=O(n^2)$ conflicts in the graph initially, we obtain a straightforward upper bound of $O(n^2)$ on the stabilization time.

In case of basic switching, a node switches its color for an arbitrarily small improvement. Alternatively, if we denote the complement of $N_c(v)$ by $N_{\overline{c}}(v):=N(v) \setminus N_c(v)$, we can also formulate this rule as $|N_c(v)| - |N_{\overline{c}}(v)| > 0$. In case of the worst possible initial coloring, this rule is known to allow a stabilization time of $\Theta(n^2)$ \cite{minority, majority, votingtime}.

In contrast to this, proportional switching is defined for a specific parameter $\lambda \in (0,1]$, and it requires that $v$ is in conflict with a specific portion of its neighborhood, with $\frac{1+\lambda}{2} \in (\frac{1}{2}, 1]$. This is often a more realistic approach if nodes have a large degree, or if switching also induces some cost in an application area. Equivalently, we can rephrase this rule as $|N_c(v)| - |N_{\overline{c}}(v)| \geq \lambda \cdot d_v$. This shows that whenever $v$ switches, the total number of conflicts in the graph decreases by at least $\lambda \cdot d_v$, and $v$ can have at most $\frac{1+\lambda}{2} \cdot d_v - \lambda \cdot d_v = \frac{1-\lambda}{2} \cdot d_v$ conflicts on the incident edges after the switch.

In case of a worst-case initial coloring, the maximal stabilization time for proportional switching is between quadratic and linear, following a monotonously decreasing non-elementary function $f(\lambda)$ described in \cite{prop}.
Since this non-elementary function also plays a role in our lower bound, we briefly discuss $f(\lambda)$ in Appendix \ref{App:C} for completeness.

Note that for a very small $\lambda$ value approaching $0$, we can obtain basic switching as a special case of proportional switching in the limit.

\subsection{Application of earlier results} \label{sec:earlier}

We also apply the basic ideas behind some of the constructions from previous work, which were used to show similar lower bounds for a worst-case initial coloring.

\subparagraph*{Construction idea for basic switching.}

Recall that the result of \cite{minority} provides a quadratic lower bound on the stabilization time of minority processes. 

\begin{theorem*}[from \cite{minority}]
Consider minority processes under the basic switching rule. There exists a class of graphs and an initial coloring with a stabilization time of $\Omega(n^2)$.
\end{theorem*}

The main idea of the construction is to have a set $P$ of $m$ nodes, attached to two further sets $A$ and $B$ of size $m$. The construction makes sure that every node in $A$ and $B$ wants to switch to the opposite color. Then we switch these nodes in an alternating fashion: one from $A$, one from $B$, one from $A$ again, and so on. The set $P$ is designed such that its neighborhood is approximately balanced, and thus after each of these steps, the entire set $P$ is switchable. Switching $P$ after each step gives a sequence of $m \cdot 2m = \Theta(n^2)$ switches.

\subparagraph*{Black box construction for proportional switching.}

We also use the result of \cite{prop}, which provides a lower bound construction for any $\lambda \leq \frac{1}{3}$ in case of proportional switching and worst-case initial coloring. We apply this graph as a black box in our constructions, and refer to it as the \textsc{prop} construction. 

\begin{theorem*}[from \cite{prop}]
Consider majority/minority processes under proportional switching for any $\lambda \leq \frac{1}{3}$. There exists a class of graphs and an initial coloring with a stabilization time of $\Omega(n^{1+f(\lambda)-\varepsilon})$ for the function $f$ described in Appendix \ref{App:C} and for any $\varepsilon>0$.
\end{theorem*}

\section{Basic observations}

\subsection{Initially balanced sets}

Since we start from a uniform random initial coloring, a basic tool in our proofs is the fact that w.h.p., a large set of nodes has a balanced distribution of the colors initially.

\begin{definition}[$\epsilon$-balanced set]
Given a specific coloring, we say that a set of nodes $S$ is \emph{$\epsilon$-balanced} if the number of white nodes in $S$ is within $\left[(\frac{1}{2}-\epsilon) \cdot |S|, \, (\frac{1}{2}+\epsilon) \cdot |S| \right]$.
\end{definition}

\begin{lemma} \label{lem:epsbalanced_initially}
Let $S_1, ..., S_k$ be subsets of nodes in $G$ such that $|S_i| \geq c_0 \cdot \log{n}$ for some constant $c_0$ for all $i \in \{1, ..., k \}$, and $k \leq n$. Then for any constant $\epsilon>0$, there is a $c_0$ such that w.h.p., each set $S_i$ is initially $\epsilon$-balanced.
\end{lemma}

\begin{proof}
Let us select $c_0=\frac{3}{\epsilon^2}$. According to the Chernoff bound, the probability that $S_i$ is not $\epsilon$-balanced is at most
\[ 2 \cdot e^{-4 \epsilon^2 \cdot \frac{1}{6} \cdot |S_i|} \, \leq \, 2 \cdot e^{-\frac{2}{3} \epsilon^2 \cdot c_0 \cdot \log{n}} \, = \,  2 \cdot n^{-2}. \]
If we take a union bound over all the $k \leq n$ subsets, the probability that any of them is not $\epsilon$-balanced is at most $n \cdot 2 \cdot n^{-2} = 2 \cdot n^{-1}$, so w.h.p. the claim indeed holds. 
\end{proof}

In particular, we can select a high constant $c_0$, and refer to nodes $v$ with $d_v \geq c_0 \cdot \log{n}$ as \textit{high-degree} nodes, and the remaining nodes as \textit{low-degree} nodes. Then Lemma \ref{lem:epsbalanced_initially} can be rephrased into the following claim:

\begin{corollary} \label{lem:epsbalanced_neighborhoods}
For any $\epsilon>0$, there exists a $c_0$ such that w.h.p. the following claim holds: for all the high-degree nodes $v$ in $G$, $N(v)$ is initially $\epsilon$-balanced.
\end{corollary}

\subsection{Linear lower bound} \label{sec:triv_lower}

Note that we can easily provide an example of linear stabilization time, even for proportional switching with any $\lambda \in (0,1)$.

Consider an edge graph, i.e.\ a connected component with only two adjacent nodes $u$ and $v$. With a probability of $\frac{1}{2}$, node $v$ is initially switchable in this graph, for both majority/minority processes (since it has the opposite/same color as $u$, respectively). Let us take $\frac{n}{2}$ independent copies of this single-edge graph; this gives $\frac{n}{2}$ nodes in the role of $v$. Then $\frac{n}{4}$ of these nodes are switchable in expectation, and with a Chernoff bound, one can show that at least $\frac{n}{8}$ are switchable w.h.p.. We can switch these $\frac{n}{8}$ nodes in any order to obtain a sequence of $\frac{n}{8} \in \Omega(n)$ switches.

\section{Lower bound constructions for basic switching} \label{sec:basic}

For basic switching, we can give an example of quadratic stabilization time by a suitable extension of the construction in \cite{minority} to the random-initialized setting.

In our analysis, we refer to a set of nodes as a \textit{group} if they all have exactly the same neighborhood. In our figures, we denote groups by double-sided circles, with the cardinality shown beside the group, and an edge between two groups denotes a complete bipartite connection between the corresponding sets. Note that the nodes of a group always prefer the same color.

\begin{theorem}
Consider majority/minority processes under the basic switching rule, starting from a uniform random initial coloring. There exists a class of graphs that exhibit a stabilization time of $\Omega(n^2)$ with high probability in this model.
\end{theorem}

\noindent We now outline the main ideas of these graphs, with the details discussed in Appendix \ref{App:A}.

\subsection{Minority processes} \label{sec:simple_min}

For minority processes, consider the graph in Figure \ref{fig:basic}, which is essentially an extension of the graph in \cite{minority} with a complete bipartite connection between $A_0$ and $B_0$. For simplicity, we add an extra node to ensure that $P$ has an odd degree. The graph has $5m+3$ nodes, and thus $m \in \Theta(n)$.

Regardless of the initial coloring, each node in $A_0$ has the same preferred color, since they all have exactly the same neighbors and they have an odd degree. Thus we can switch each node in $A_0$ to this preferred color (if it did not have this color already). Assume w.l.o.g. that this color is white. Since now $A_0$ is white entirely, we can switch each node in $B_0$ to black. With this, the preferred color of each node in $A$ becomes black, and the preferred color of each node in $B$ becomes white.

An intuitive description of the remaining sequence is as follows. Both $A$ and $B$ have approximately $\frac{m}{2}$ nodes (and w.h.p. at least $\frac{m}{3}$ nodes) that have the same color as the group above. These nodes are now all switchable, regardless of the color of nodes in $P$. We disregard the remaining nodes, and only focus on these $\frac{m}{3}$ switchable nodes in $A$ and $B$.

Initially, the neighborhood of $P$ is w.h.p. $\epsilon$-balanced. Hence by switching only $\epsilon \cdot m$ of nodes either in $A$ or in $B$, we can ensure that $P$ has exactly one more white neighbor than black, which allows us to switch the entire group $P$ to black. Then by switching one node in $A$ to black, $P$ will have one more black neighbor than white, so $P$ becomes switchable again. We can then switch the nodes in $A$ and $B$ in an alternating fashion; this ensures that $P$ always has one more same-colored neighbor after each step, which makes $P$ switchable again. This process allows us to switch the nodes of $P$ altogether $\Theta(m)$ times, which already adds up to a sequence of $\Theta(m^2)=\Theta(n^2)$ switches.

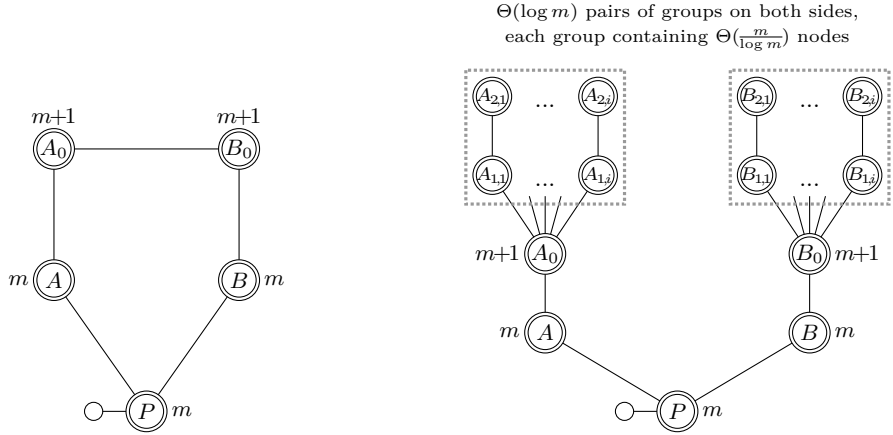
\begin{figure}
\centering
\hspace{0.04\textwidth}
\begin{subfigure}[b]{0.4\textwidth}
	\scalebox{.99}{\begin{tikzpicture}
	
	\draw (35pt,0pt) -- (0pt,50pt);
	\draw (35pt,0pt) -- (70pt,50pt);
	\draw (0pt,50pt) -- (0pt,100pt);
	\draw (70pt,50pt) -- (70pt,100pt);
	\draw (0pt,100pt) -- (70pt,100pt);
	
	\draw (15pt,0pt) -- (35pt,0pt);
	\draw[black, fill=white] (15pt,0pt) circle (0.8ex);
	
	\draw[black, fill=white] (35pt,0pt) circle (1.8ex);
	\draw[black, fill=white] (35pt,0pt) circle (1.5ex);
	
	\draw[black, fill=white] (0pt,50pt) circle (1.8ex);
	\draw[black, fill=white] (0pt,50pt) circle (1.5ex);
	
	\draw[black, fill=white] (70pt,50pt) circle (1.8ex);
	\draw[black, fill=white] (70pt,50pt) circle (1.5ex);
	
	\draw[black, fill=white] (0pt,100pt) circle (1.8ex);
	\draw[black, fill=white] (0pt,100pt) circle (1.5ex);
	
	\draw[black, fill=white] (70pt,100pt) circle (1.8ex);
	\draw[black, fill=white] (70pt,100pt) circle (1.5ex);
	
	\node[anchor=south] at (0pt,105.5pt) {\footnotesize $m\!\!+\!\!1$};
	\node[anchor=south] at (70pt,105.5pt) {\footnotesize $m\!\!+\!\!1$};
	\node[anchor=east] at (-6pt,50pt) {\footnotesize $m$};
	\node[anchor=west] at (76pt,50pt) {\footnotesize $m$};
	\node[anchor=west] at (41pt,0pt) {\footnotesize $m$};
	
	\node[anchor=center] at (35pt,0pt) {\footnotesize $P$};
	\node[anchor=center] at (0pt,50pt) {\footnotesize $A$};
	\node[anchor=center] at (70pt,50pt) {\footnotesize $B$};
	\node[anchor=center] at (0pt,100pt) {\footnotesize $A_0$};
	\node[anchor=center] at (70pt,100pt) {\footnotesize $B_0$};
	
\end{tikzpicture}}
\end{subfigure}
\hspace{0.02\textwidth}
\begin{subfigure}[b]{0.5\textwidth} 
	\scalebox{.99}{\begin{tikzpicture}
	
	\draw (50pt,0pt) -- (0pt,30pt);
	\draw (50pt,0pt) -- (100pt,30pt);
	\draw (0pt,30pt) -- (0pt,60pt);
	\draw (100pt,30pt) -- (100pt,60pt);
	
	\draw (0pt,60pt) -- (20pt,90pt);
	\draw (0pt,60pt) -- (-6pt,82pt);
	\draw (0pt,60pt) -- (0pt,82pt);
	\draw (0pt,60pt) -- (6pt,82pt);
	\draw (0pt,60pt) -- (-20pt,90pt);
	
	\draw (20pt,90pt) -- (20pt,120pt);
	\draw (-20pt,90pt) -- (-20pt,120pt);
	
	\draw (100pt,60pt) -- (120pt,90pt);
	\draw (100pt,60pt) -- (94pt,82pt);
	\draw (100pt,60pt) -- (100pt,82pt);
	\draw (100pt,60pt) -- (106pt,82pt);
	\draw (100pt,60pt) -- (80pt,90pt);
	
	\draw (120pt,90pt) -- (120pt,120pt);
	\draw (80pt,90pt) -- (80pt,120pt);
	
	\draw (30pt,0pt) -- (50pt,0pt);
	\draw[black, fill=white] (30pt,0pt) circle (0.8ex);
	
	\draw[black, fill=white] (50pt,0pt) circle (1.8ex);
	\draw[black, fill=white] (50pt,0pt) circle (1.5ex);
	
	\draw[black, fill=white] (0pt,30pt) circle (1.8ex);
	\draw[black, fill=white] (0pt,30pt) circle (1.5ex);
	
	\draw[black, fill=white] (100pt,30pt) circle (1.8ex);
	\draw[black, fill=white] (100pt,30pt) circle (1.5ex);
	
	\draw[black, fill=white] (0pt,60pt) circle (1.8ex);
	\draw[black, fill=white] (0pt,60pt) circle (1.5ex);
	
	\draw[black, fill=white] (100pt,60pt) circle (1.8ex);
	\draw[black, fill=white] (100pt,60pt) circle (1.5ex);
	
	\draw[black, fill=white] (-20pt,90pt) circle (1.7ex);
	\draw[black, fill=white] (-20pt,90pt) circle (1.4ex);
	
	\node[anchor=center] at (0pt,86.5pt) {\small ...};
	
	\draw[black, fill=white] (20pt,90pt) circle (1.7ex);
	\draw[black, fill=white] (20pt,90pt) circle (1.4ex);
	
	\draw[black, fill=white] (-20pt,120pt) circle (1.7ex);
	\draw[black, fill=white] (-20pt,120pt) circle (1.4ex);
	
	\node[anchor=center] at (0pt,116.5pt) {\small ...};
	
	\draw[black, fill=white] (20pt,120pt) circle (1.7ex);
	\draw[black, fill=white] (20pt,120pt) circle (1.4ex);
	
	
	\draw[black, fill=white] (80pt,90pt) circle (1.7ex);
	\draw[black, fill=white] (80pt,90pt) circle (1.4ex);
	
	\node[anchor=center] at (100pt,86.5pt) {\small ...};
	
	\draw[black, fill=white] (120pt,90pt) circle (1.7ex);
	\draw[black, fill=white] (120pt,90pt) circle (1.4ex);
	
	\draw[black, fill=white] (80pt,120pt) circle (1.7ex);
	\draw[black, fill=white] (80pt,120pt) circle (1.4ex);
	
	\node[anchor=center] at (100pt,116.5pt) {\small ...};
	
	\draw[black, fill=white] (120pt,120pt) circle (1.7ex);
	\draw[black, fill=white] (120pt,120pt) circle (1.4ex);
	
	\node[anchor=east] at (-6pt,60pt) {\footnotesize $m\!\!+\!\!1$};
	\node[anchor=west] at (106pt,60pt) {\footnotesize $m\!\!+\!\!1$};
	\node[anchor=east] at (-6pt,30pt) {\footnotesize $m$};
	\node[anchor=west] at (106pt,30pt) {\footnotesize $m$};
	\node[anchor=west] at (56pt,0pt) {\footnotesize $m$};
	
	
	
	\node[anchor=center] at (50pt,0pt) {\footnotesize $P$};
	\node[anchor=center] at (0pt,30pt) {\footnotesize $A$};
	\node[anchor=center] at (100pt,30pt) {\footnotesize $B$};
	\node[anchor=center] at (0pt,60pt) {\footnotesize $A_0$};
	\node[anchor=center] at (100pt,60pt) {\footnotesize $B_0$};
	
	\node[anchor=center] at (-20pt,89.5pt) {\scriptsize $A_{_{\!}1\!,\!1}$};
	\node[anchor=center] at (20pt,89.5pt) {\scriptsize $A_{_{\!}1\!,\!i}$};
	
	\node[anchor=center] at (-20pt,119.5pt) {\scriptsize $A_{_{\!}2\!,\!1}$};
	\node[anchor=center] at (20pt,119.5pt) {\scriptsize $A_{_{\!}2\!,\!i}$};
	
	\node[anchor=center] at (80pt,89.5pt) {\scriptsize $B_{_{\!}1\!,\!1}$};
	\node[anchor=center] at (120pt,89.5pt) {\scriptsize $B_{_{\!}1\!,\!i}$};
	
	\node[anchor=center] at (80pt,119.5pt) {\scriptsize $B_{_{\!}2\!,\!1}$};
	\node[anchor=center] at (120pt,119.5pt) {\scriptsize $B_{_{\!}2\!,\!i}$};
	
	
	\draw[very thick, lgray, densely dotted] (-30pt,130pt) -- (-30pt,79pt) -- (30pt,79pt) -- (30pt,130pt) -- cycle;
	
	\draw[very thick, lgray, densely dotted] (70pt,130pt) -- (70pt,79pt) -- (130pt,79pt) -- (130pt,130pt) -- cycle;
	
	
	\node[anchor=south] at (50pt,145pt) {\scriptsize $\Theta(\log{m})$ pairs of groups on both sides,};
	\node[anchor=south] at (50pt,134pt) {\scriptsize each group containing $\Theta(\!\frac{m}{\log{m}}\!)$ nodes};
	
\end{tikzpicture}}
\end{subfigure}
\caption{Lower bound constructions of $\Omega(n^2)$ steps in case of basic switching, for minority processes (left) and majority processes (right). Recall that double-sided circles denote groups, and edges between groups denote a complete bipartite connection between the two groups.}
\label{fig:basic}
\end{figure}

\subsection{Majority processes} \label{sec:simple_maj}

The case of majority processes is more involved, since in this case, it is more difficult to ensure that the groups $A_0$ and $B_0$ attain different colors.

Instead of connecting $A_0$ to $B_0$, we connect $A_0$ to $\Theta(\log{m})$ further groups of size $\Theta(\frac{m}{\log{m}})$, denoted by $A_{1, 1}$, $A_{1, 2}$, $...$ . Finally, we add $\Theta(\log{m})$ more distinct groups $A_{2, 1}$, $A_{2, 2}$, $...$, also on $\Theta(\frac{m}{\log{m}})$ nodes each, and we create a complete bipartite connection between $A_{1, i}$ and $A_{2,i}$. We attach the same structures to group $B_0$ in a symmetric manner; see Figure \ref{fig:basic} for an overview of the construction.

The main idea of the construction is as follows. With probability $\frac{1}{2}$, the group $A_{1,i}$ has more white nodes than black initially, which allows us to switch $A_{2, i}$ entirely to white. Since the groups $A_{1,i}$ are independent, there is indeed w.h.p. an index $\hat{i}$ such that the group $A_{2,\hat{i}}$ can be switched entirely to white. The neighbors of $A_{1,\hat{i}}$ are initially approximately balanced, so after recoloring all the $\Theta(\frac{m}{\log{m}})$ nodes in $A_{2, \hat{i}}$ to white, $A_{1,\hat{i}}$ has more white neighbors than black; this allows us to switch all of $A_{1,\hat{i}}$ to white. We note while our previous steps all follow directly from Corollary \ref{lem:epsbalanced_neighborhoods}, this specific step requires a slightly stronger version of the Chernoff bound.

We can then apply a similar reasoning on the group $A_0$: since it was w.h.p. balanced initially, and turning $A_{1,\hat{i}}$ to white has increased the number of its white neighbors by $\Theta(\frac{m}{\log{m}})$ w.h.p., we can also turn the entire group $A_0$ white. In a similar fashion, we can use groups $B_{2,\hat{i}}$ and $B_{1,\hat{i}}$ to switch each node in $B_0$ black w.h.p..

Once $A_0$ is white and $B_0$ is black, we again have $\Theta(m)$ switchable nodes in both $A$ and $B$, and thus we can apply the same alternating method as in the minority case.

\section{Proportional switching: lower bound for $\lambda < \frac{1}{3}$} \label{sec:lower}

We now show that for proportional switching with small $\lambda$ values, stabilization time can indeed be superlinear. Note that $\lambda < \frac{1}{3}$ implies that $\frac{1+\lambda}{2}=\frac{2}{3}-\delta$ for some $\delta>0$.

We present our lower bound construction for majority processes; however, since our graph is bipartite, we can easily adapt this result to minority processes by inverting the colors in one of the color classes. More details of this technique are available in Appendix \ref{App:B}.

\begin{theorem} \label{th:lower}
Consider majority/minority processes under the proportional switching rule for any $\lambda < \frac{1}{3}$, starting from a uniform random initial coloring. For any $\varepsilon>0$, there exists a class of graphs that exhibit a stabilization time of $\Omega\left(n^{1+f\left(\frac{2 \cdot \lambda}{1-\lambda}\right) -\varepsilon}\right)$ with high probability.
\end{theorem}

\noindent In a simplified formulation, this means that there exists a constant $c>0$ such that there is a construction with a stabilization time of $\Omega(n^{1+c})$ in this setting.

We divide our construction technique into five main \textit{phases}, and discuss them separately. In each phase of the construction, we will refer to some edges of the nodes as \textit{output} edges, which go to the following phase of the construction. In a specific phase, we always achieve a desired behavior without any change on these output neighbors yet. An overview of the entire construction is available in Figure \ref{fig:prop}.

As before, we define our construction in terms of a parameter $m=\widetilde{\Theta}(n)$, and discuss the value of $m$ in the end.

\begin{figure}
\centering
	\scalebox{.99}{\input{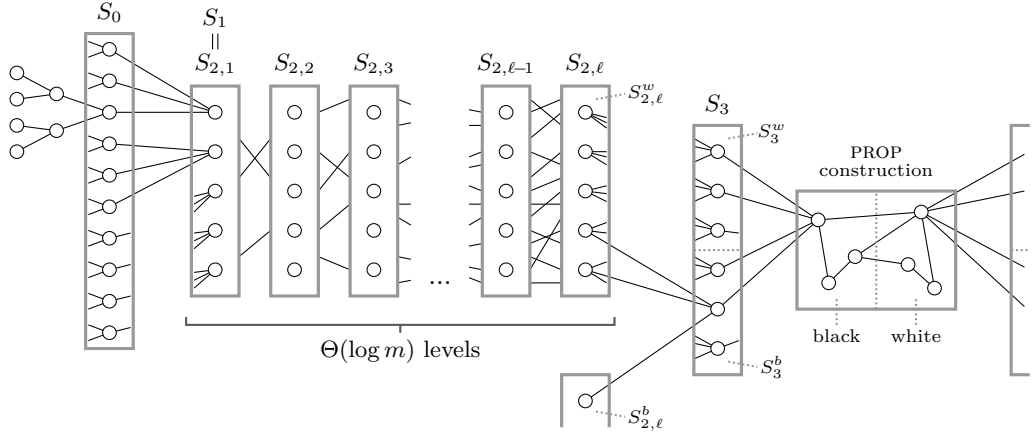}}
	\caption{High-level illustration of the proportional lower bound construction for any $\lambda< \frac{1}{3}$.}
	\label{fig:prop}
\end{figure}

\begin{itemize}

\setlength\topsep{2pt}
\setlength\itemsep{4pt}

\item First, in the \textit{Opening Phase}, our goal is to create a set $S_0$ of constant-degree nodes such that (i) each node in $S_0$ has $1$ output edge to the next phase, and (ii) for any parameter $p<1$, we can switch each node in $S_0$ to black with a probability of at least $p$, independently from the remaining nodes.

\item In the \textit{Collection Phase}, we use our Opening Phase construction to produce another set $S_1$ where (i) each node in $S_1$ has $c_0 \cdot \log{n}$ output edges for a large enough constant $c_0$, and (ii) w.h.p. we can switch all the nodes in $S_1$ to black.

\item In the \textit{Growing Phase}, we begin with this node set $S_{2,1}:=S_1$, and add a range of further \textit{levels} $S_{2,2}$, $S_{2,3}$, ... of the same size. Every level $S_{2,i}$ is only connected to the previous and next levels $S_{2,i-1}$ and $S_{2,i+1}$. The levels will have an exponentially increasing output degree, and hence in at most $\ell \approx \log{m}$ steps, we arrive at a final level $S_{2,\ell}$ where each node has an output degree of $\Theta(m)$. As in case of $S_1$, we show that we can w.h.p. turn each node in $S_{2,i}$ (and finally, in $S_{2,\ell}$) black.

\item In the \textit{Control Phase}, we use $S_{2,\ell}$ to produce a set $S_3$ where each node still has an output degree of $\Theta(m)$. We will ensure that (i) there is a specific point in the process where each node in $S_3$ is switchable to black, and (ii) later, there is a specific point in the process where each node in $S_3$ is switchable to white.

\item Finally, in the \textit{Simulation Phase}, we take an instance of the \textsc{prop} construction, and we use our set $S_3$ to force each node in this construction to take the desired ``initial'' color. We can then simulate the behavior of \textsc{prop} as a black box, which is known to provide a superlinear stabilization time from this artificially enforced worst-case initial coloring.

\end{itemize}

\noindent In this section, we outline the main ideas behind each of these phase. More details of the construction are discussed in Appendix \ref{App:B}.

We note that the second and third phases can be generalized to any $\lambda$ up to $\frac{1}{2}$; however, there is no straightforward way to do this for the remaining phases.

\subsection{Opening Phase} \label{sec:opening}

To construct the set $S_0$, first consider a node $v$ with $d_{v}=3$: one neighbor labeled as an output, and two further neighbors $u_1$ and $u_2$. Initially, we have an $\frac{1}{2}$ chance that $v$ is already black. Even if $v$ is not black initially, we can switch it black if both $u_1$ and $u_2$ are black initially: we have $\frac{1+\lambda}{2}<\frac{2}{3}$, so 2 black neighbors out of 3 are indeed enough to make $v$ switchable. The probability that initially $v$ is white but $u_1$ and $u_2$ are black is $\left( \frac{1}{2} \right)^3=\frac{1}{8}$, so altogether, we can turn $v$ black with a probability of $p_1=\frac{5}{8}$.

Now assume that we take two such nodes that can be switched black with probability $\frac{5}{8}$, we denote them by $u'_1$ and $u'_2$, and we connect their outputs to a new node $v'$. Again, $v'$ is already black initially with probability $\frac{1}{2}$; if not, we can turn $v'$ black if both $u'_1$ and $u'_2$ are switched black, which happens with a probability of $p_1^{\ 2}$. This provides a black $v'$ with a probability of $p_2=\frac{1}{2}+\frac{1}{2} \cdot \left( \frac{5}{8} \right)^2=\frac{89}{128}$.

We can continue this in a recursive manner, always taking two copies of the previous construction, and connecting them to a new root node. After $i$ steps, we end up with a full binary tree on $2^{i+1}-1$ nodes. This provides a black root node with a probability of $p_i$, defined by the recurrence 
\[p_0=\frac{1}{2} \: \quad \text{ and } \: \quad p_{i+1} = \frac{1}{2}+\frac{1}{2} \cdot p_i^{\ 2}.\]
One can easily show that $\lim_{i \rightarrow \infty} \, p_i=1$. Hence for any constant parameter $p<1$, there is an $i$ such that $p_i \geq p$, and thus creating $i$ layers with this method ensures that we can switch the final node black with probability at least $p$.

In order to build our set $S_0$, we can simply take $m_0=|S_0|$ independent copies of this tree. Since $p$ is a constant, $i$ and the tree size $2^{i+1}-1$ are also constants; thus the whole phase only requires $O(m_0)$ nodes altogether.

\subsection{Collection Phase}

Let us introduce a logarithmic parameter $d_0=c_0 \cdot \log{n}$. Given our Opening Phase construction $S_0$, our next step is to create a smaller set $S_1$ on $m_1 = \frac{1}{4 \cdot d_0} \cdot m_0$ nodes. Recall that all the $m_0$ nodes in $S_0$ had exactly $1$ output edge; this allows us to connect each $v \in S_1$ to $4 \cdot d_0$ distinct nodes in $S_0$. We also add $d_0$ further output edges to each $v \in S_1$ to provide a connection to the next phase.

Since each node in $S_0$ becomes black with probability $p$ independently, a Chernoff bound shows that $v$ has at least $(p-\epsilon) \cdot 4 \cdot d_0$ black neighbors in $S_0$ with a probability of $1-O(n^{-2})$. This already makes $v$ switchable to black, since $d_v=5 \cdot d_0$, and thus for the appropriate $p$ and $\epsilon$ values we have 
\[ \frac{(p-\epsilon) \cdot 4 \cdot d_0}{5 \cdot d_0} \, \approx \, \frac{4}{5} \, > \, \frac{2}{3} > \frac{1+\lambda}{2} \,. \]
Applying a union bound over all nodes $v \in S_1$, we get that w.h.p. the entire set $S_1$ can be switched to black.

\subsection{Growing Phase}

Given our set $S_1$ from the Collection Phase, the next step is to iteratively build a range of levels $S_{2,i}$ for $i=1,2,...$ . Each of these levels has the same size $|S_{2,i}|=m_1$, but on the other hand, their degrees increase exponentially: the output degree of each node in $S_{2,i+1}$ is always twice as big as the output degree of the nodes in $S_{2,i}$.

We achieve this by connecting every pair of subsequent levels as a regular bipartite graph. Let us begin with $S_{2,1}:=S_1$. Recall that each node in $S_1$ has $d_0$ output edges, so $S_{2,1}$ and $S_{2,2}$ will form a $d_0$-regular bipartite graph. We then connect $S_{2,2}$ and $S_{2,3}$ as a $2 \cdot d_0$-regular bipartite graph, $S_{2,3}$ and $S_{2,4}$ as a $4 \cdot d_0$-regular bipartite graph, and so on. Thus in any level, we have a value $d$ such that each node has $d$ edges to the previous and $2 d$ edges to the next level, and this value $d$ doubles with each new level. Since the degrees grow exponentially, after about $\log{m_1}$ levels, we reach a last level $S_{2,\ell}$ where the output degree is $\Theta(m_1)$.

We use an induction to prove that we can w.h.p. turn all nodes black in each $S_{2,i}$. This is already known for $S_{2,1}=S_1$ initially. In the general case, let $v$ be an arbitrary node of $S_{2,i}$. Since each $v$ has at least $d_0$ output edges to $S_{2,i+1}$, we can use Lemma \ref{lem:epsbalanced_initially} to show that the output neighborhood of every node is initially $\epsilon$-balanced. This means that for any $v \in S_{2,i}$, at least $(\frac{1}{2}-\epsilon) \cdot 2 d = (1-2\epsilon) \cdot d$ outputs are already black initially. Due to the induction, we can turn all the $d$ remaining neighbors in $S_{2,i-1}$ black, altogether giving $(2 - 2\epsilon) \cdot d$ black neighbors of $v$. With $d_v=3 \cdot d$, this amounts to a ratio of $\frac{2-2\epsilon}{3}$ black nodes in $N(v)$. Since we have $\frac{1+\lambda}{2} = \frac{2}{3} - \delta$, a sufficiently small choice of $\epsilon$ always ensures that this ratio is above $\frac{1+\lambda}{2}$, and thus $v$ is switchable to black. Hence each node in $S_{2,i}$ can indeed be turned black, which completes our induction.

\subsection{Control Phase}

In the following Control Phase, we create a new set $S_3$ on $m_3$ nodes. The goal of this phase is to ensure that at a specific point in the process, each $v \in S_3$ switches to black, and then at a later point, each $v \in S_3$ is switchable to white.

In order to be able to initialize a \textsc{prop} construction on $m$ nodes in the final phase, each node in $S_3$ will have an output degree of $m$, for some parameter $m$. A detailed analysis shows that for a large constant $\alpha>1$, a choice of $m_3 = \frac{1}{\alpha} \cdot m_1$ and $m = \frac{1}{2} \cdot m_3$ suffices for our purposes.

To achieve the desired switching behavior for $S_3$, we first create two copies of the previous phases: one of them ending with a level $S_{2,\ell}^{b}$ on $\alpha \! \cdot \! m$ nodes where w.h.p. each nodes switches to black, and the other one ending with a last level $S_{2,\ell}^{w}$ on $2\alpha \! \cdot \! m$ nodes where w.h.p. each node switches to white in a symmetric manner. We connect each node in $S_3$ to every node in both $S_{2,\ell}^{b}$ and $S_{2,\ell}^{w}$. As a result, each $v \in S_3$ has a degree of $d_v=(3 \alpha \! + \! 1) \! \cdot \! m$. Note that the output degree of both $S_{2,\ell}^{b}$ and $S_{2,\ell}^{w}$ is $\Theta(m_1) = \Theta(\alpha \cdot m_3)$, so for $\alpha$ large enough, they can indeed be connected to each node in $S_3$.

Now consider the neighbors of a node $v \in S_3$. First $S_{2,\ell}^{b}$ becomes black and $v$'s neighborhood in $S_{2,\ell}^{w}$ is $\epsilon$-balanced; this gives at least $\alpha \cdot m + (\frac{1}{2}-\epsilon) \cdot 2\alpha \cdot m = 2\alpha \cdot m \cdot (1-\epsilon)$ black neighbors in $N(v)$, amounting to a $\frac{2 \alpha \cdot (1-\epsilon)}{3\alpha+1}$ fraction of $d_v$. As $\frac{1+\lambda}{2} = \frac{2}{3} - \delta$, for a sufficiently small $\epsilon$ and sufficiently large $\alpha$, we can ensure that this ratio is larger than $\frac{1+\lambda}{2}$, and thus $v$ is indeed switchable. We switch each $v \in S_3$ to black at this point.

After this, we turn each node in $S_{2,\ell}^{w}$ white. Nodes in $S_3$ now have $2\alpha \cdot m$ white neighbors at least; this again ensures that each $v \in S_3$ is now switchable to white. However, for our purposes in the last phase, we will only switch half of the nodes in $S_3$ white at this point (denoted by $S_3^{w}$), and leave the remaining part black (denoted by $S_3^{b}$).

\subsection{Simulation Phase}

Finally, we use the \textsc{prop} construction on $m$ nodes to obtain superlinear stabilization time. Given a node $v$ in \textsc{prop}, assume w.l.o.g. that $v$ is initially black in the example sequence of \textsc{prop}; we can apply the same technique for white nodes in a symmetric manner.

Our main idea is to connect $v$ to some new nodes in $S_3^{b}$ and $S_3^{w}$. When $S_3^{b}$ and $S_3^{w}$ both switch to black, this allows us to switch $v$ to its desired initial color (black). Then when $S_3^{w}$ switches back to white, the new neighbors become balanced, and thus the switchability of $v$ will again depend on its original neighbors within \textsc{prop}. However, with these extra connections, the original $N(v)$ is now only a smaller fraction of $v$'s total neighborhood, so this only allows us to simulate \textsc{prop} with a smaller parameter $\lambda'<\lambda$.

More specifically, if $v$ has original degree $d_v'$ within the \textsc{prop} construction, then we connect $v$ to $\frac{1}{2} \cdot \frac{1+\lambda}{1-\lambda} \cdot d_v'$ arbitrary nodes in both $S_3^{b}$ and $S_3^{w}$. We point out that our choice of $m= \frac{1}{2} \cdot m_3$ is indeed sufficient for this: since $\lambda < \frac{1}{3}$ implies $\frac{1+\lambda}{1-\lambda} < 2$, every node in the \textsc{prop} construction needs at most $\frac{1}{2} \cdot \frac{1+\lambda}{1-\lambda} \cdot d'_v \, < \, d'_v$ new edges to both $S_3^b$ and $S_3^w$. Hence with $d'_v<m$ in the \textsc{prop} construction, it is indeed enough to have $m$ nodes in the sets $S_3^b$ and $S_3^w$. Furthermore, since each node in $S_3$ has an output degree of $m$, we can also connect a node in $S_3^b$ or $S_3^w$ to as many nodes in the \textsc{prop} construction as necessary.

With $v$ connected to $\frac{1}{2} \cdot \frac{1+\lambda}{1-\lambda} \cdot d_v'$ nodes in both $S_3^{b}$ and $S_3^{w}$, the new degree of $v$ is now
\[ d_v= \left( 1 + \frac{1+\lambda}{1-\lambda} \right) \cdot d_v' = \frac{2}{1-\lambda} \cdot d_v' \, , \]
so $v$ requires $\frac{1+\lambda}{2} \cdot d_v = \frac{1+\lambda}{1-\lambda} \cdot d_v'$ conflicts to be switchable. Hence when $S_3^{b}$ and $S_3^{w}$ are both switched black, this is already enough to switch $v$ black, since the two sets provide $2 \cdot \frac{1}{2} \cdot \frac{1+\lambda}{1-\lambda} \cdot d_v' = \frac{1+\lambda}{1-\lambda} \cdot d_v'$ black neighbors to $v$ together. Later $S_3^{w}$ switches to white; then for the rest of the process, $v$ has $\frac{1}{2} \cdot \frac{1+\lambda}{1-\lambda} \cdot d_v'$ neighbors of both colors in $S_3$.

Let us now select $\lambda'=\frac{2 \lambda}{1- \lambda}$, and apply the \textsc{prop} construction for $\lambda'$ as a black box. If $v$ was switchable in the original \textsc{prop} construction at some point, then it had at least $\frac{1+\lambda'}{2} \cdot d_v' \, = \, \frac{1}{2} \cdot \frac{1+\lambda}{1-\lambda} \cdot d_v'$ conflicts within \textsc{prop}. Then together with the $\frac{1}{2} \cdot \frac{1+\lambda}{1-\lambda} \cdot d_v'$ additional conflicts to either $S_3^{b}$ or $S_3^{w}$, $v$ has at least $\frac{1+\lambda}{1-\lambda} \cdot d_v' = \frac{1+\lambda}{2} \cdot d_v$ conflicts in our construction, and thus it is indeed switchable.

Hence we can indeed simulate the behavior of \textsc{prop} in our construction: whenever $v$ is switchable in the original \textsc{prop} graph, it is also switchable in our construction. This allows us to run the entire sequence of $m^{1+f(\lambda') -\varepsilon}$ steps in \textsc{prop}, giving a sequence of $m^{1+f( \frac{2 \lambda}{1-\lambda} ) -\varepsilon}$ steps in terms of our $\lambda$.

One can observe that our constructions contains only $O(m \cdot \log{m})$ nodes altogether, thus allowing a choice of $m=\Theta(\frac{n}{\log{n}})$. This results in about
\[ n^{1+f\left( \frac{2 \lambda}{1-\lambda} \right) -\varepsilon} \, \cdot \, \log{n}^{\,-(1+f\left( \frac{2 \lambda}{1-\lambda} \right) -\varepsilon)} \]
steps for the \textsc{prop} sequence in terms of $n$. Since such a \textsc{prop} construction exists for any $\varepsilon>0$, we can get rid of the second factor in this lower bound by simply applying the same proof with a smaller value $\hat{\varepsilon}<\varepsilon$. Thus the claim of Theorem \ref{th:lower} follows.

\section{Proportional switching: upper bound for $\lambda > \frac{1}{2}$} \label{sec:upper}

We now show that with $\lambda = \frac{1}{2} + \delta$ for some $\delta > 0$, stabilization happens w.h.p. in $\widetilde{O}(n)$ time. The only probabilistic element of this proof is the assumption that initially all high-degree nodes have an $\epsilon$-balanced neighborhood; this indeed holds w.h.p., as we have seen before in Corollary \ref{lem:epsbalanced_neighborhoods}.

The idea of the proof is that even though there might be $\Theta(n^2)$ conflicts in the graph initially, only a few of these conflicts can propagate through the graph. Let us call a conflict on edge $(u,v)$ in our current coloring an \textit{original conflict} if it has been on the edge since the beginning of the process, i.e. if every previous state (including the initial state) already had a conflict on $(u,v)$.

\begin{definition}[Active/Rigid conflicts]
We say that a conflict on edge $(u,v)$ is \emph{rigid} if it is an original conflict, and both $u$ and $v$ are high-degree nodes. Otherwise, the conflict is \emph{active}.
\end{definition}

Our proof is obtained as a result of three observations: that (i) there are only a few active conflicts in the graph initially, (ii) the number of active conflicts decreases in each step of the process, and (iii) the process stabilizes when there are no more active conflicts. Since the second point is the most complex out of the three claims, we first discuss it separately.

\begin{lemma} \label{active_decrease}
The number of active conflicts strictly decreases in each step.
\end{lemma}

\begin{proof}
Consider a specific step of the process, and let $v$ be the node that switches in this step. Assume first that $v$ is a low-degree node. In this case, $v$ can only have active conflicts on its incident edges at any point in the process: initially, all conflicts of $v$ are active by definition, and all the newly created conflicts in the process are also active. Since the number of conflicts on $v$'s incident edges decreases when $v$ switches, the total number of active conflicts also decreases in this step.

Now assume that $v$ is a high-degree node. Since $N(v)$ is initially $\epsilon$-balanced, it has at most $(\frac{1}{2}+\epsilon) \cdot d_v$ rigid conflicts in the beginning, and since all the newly created conflicts in the process are active, it also has at most $(\frac{1}{2}+\epsilon) \cdot d_v$ rigid conflicts at any later point in the process. However, if $v$ switches, then it must have at least $\frac{1+\lambda}{2} \cdot d_v$ incident conflicts; this implies that at least $\frac{1+\lambda}{2} \cdot d_v - (\frac{1}{2}+\epsilon) \cdot d_v$ of these conflicts are active. When $v$ switches, it creates at most $\frac{1-\lambda}{2} \cdot d_v$ new (active) conflicts. Thus, to show that the number of active conflicts decreases, we only require
\[ \frac{1+\lambda}{2} \cdot d_v - \left( \frac{1}{2}+\epsilon \right) \cdot d_v > \frac{1-\lambda}{2} \cdot d_v, \]
which is equivalent to $\lambda > \frac{1}{2} + \epsilon$. This holds for a sufficiently small choice of $\epsilon < \delta$.
\end{proof}

\noindent This already allows us to prove our upper bound.

\begin{theorem}
Consider majority/minority processes under the proportional switching rule for any $\lambda > \frac{1}{2}$, starting from a uniform random initial coloring. Any graph has a stabilization time of $O(n \cdot \log{n})$ with high probability in this model.
\end{theorem}

\begin{proof}

In any initial coloring, the number of active conflicts is at most $O(n \cdot \log{n})$: each low-degree node has at most $c_0 \cdot \log{n}$ incident edges, and the number of low-degree nodes is at most $n$. Lemma \ref{active_decrease} shows that the number of active conflicts decreases in each step, so there are no active conflicts in the graph after at most $O(n \cdot \log{n})$ steps.

Once there are no more active conflicts, the coloring is stable, since nodes cannot be switchable without an active conflict on the incident edges. More specifically, due to the $\epsilon$-balanced property, all high-degree nodes $v$ have at most $(\frac{1}{2}+\epsilon) \cdot d_v$ rigid conflicts on the incident edges, which is smaller than $\frac{1+\lambda}{2} \cdot d_v$ if we have $\epsilon < \frac{\lambda}{2}$. Low-degree nodes, on the other hand, can never have rigid conflicts on the incident edges at all. Thus the process indeed stabilizes in $O(n \cdot \log{n})$ steps.
\end{proof}

\section{Conclusion}

Our results show that the behavior of the processes from a randomized initial coloring is rather straightforward in case of the basic switching rule: stabilization time can indeed tightly match the naive upper bound of $O(n^2)$.

However, in case of proportional switching, our work does leave some open questions. Figure \ref{fig:func} illustrates our upper and lower bounds for this case. The most apparent open question is the behavior of the process for the $\lambda \in [\frac{1}{3}, \frac{1}{2}]$ case; in this interval, we only have the straightforward lower bound of Section \ref{sec:triv_lower}. While the figure gives the impression that stabilization time might also have a $\widetilde{O}(n)$ upper bound in this case, it remains for future work to prove or disprove this claim.

Furthermore, even for $\lambda < \frac{1}{3}$ when stabilization is known to be superlinear, one might also be interested in devising upper bounds. Currently, the best known upper bound is that of $O(n^{1+f(\lambda)+\varepsilon})$ from \cite{prop}, which even applies for the worst-case initial coloring.

\begin{figure}
\centering
	\includegraphics[width=0.63\textwidth]{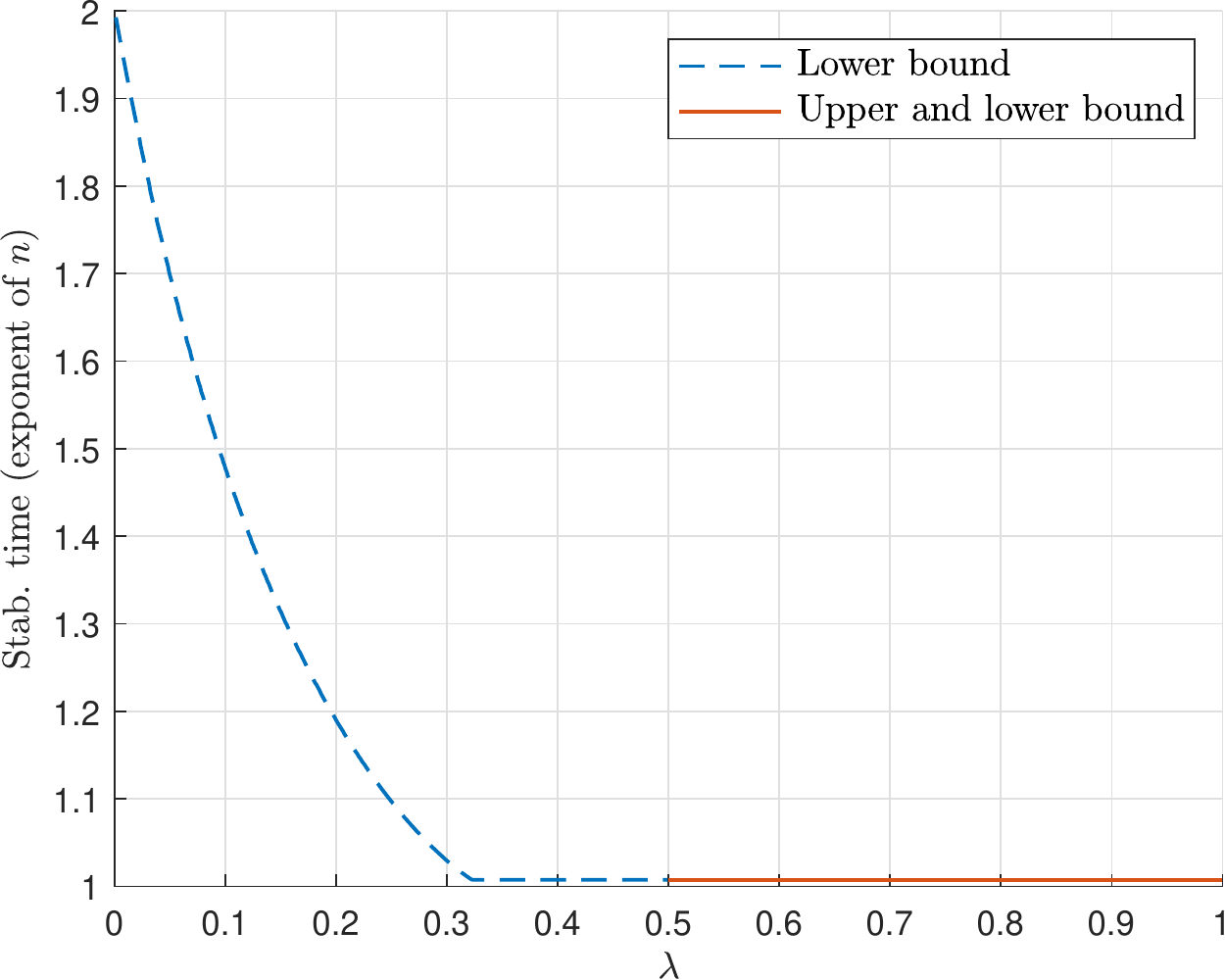}
	\caption{Our upper and lower bounds on stabilization time in the proportional case.}
	\label{fig:func}
\end{figure}

\newpage

\bibliography{references}

\begin{thebibliography}{10}

\bibitem{SocialGen2}
Victor Amelkin, Francesco Bullo, and Ambuj~K Singh.
\newblock Polar opinion dynamics in social networks.
\newblock {\em IEEE Transactions on Automatic Control}, 62(11):5650--5665,
  2017.

\bibitem{SocialGen1}
Vincenzo Auletta, Ioannis Caragiannis, Diodato Ferraioli, Clemente Galdi, and
  Giuseppe Persiano.
\newblock Generalized discrete preference games.
\newblock In {\em Proceedings of the Twenty-Fifth International Joint
  Conference on Artificial Intelligence}, IJCAI’16, page 53–59. AAAI Press,
  2016.

\bibitem{extra1}
Vincenzo Auletta, Diodato Ferraioli, and Gianluigi Greco.
\newblock On the complexity of reasoning about opinion diffusion under majority
  dynamics.
\newblock {\em Artificial Intelligence}, 284:103288, 2020.

\bibitem{approx0}
Cristina Bazgan, Zsolt Tuza, and Daniel Vanderpooten.
\newblock Complexity and approximation of satisfactory partition problems.
\newblock In {\em International Computing and Combinatorics Conference}, pages
  829--838. Springer, 2005.

\bibitem{SGPsurvey}
Cristina Bazgan, Zsolt Tuza, and Daniel Vanderpooten.
\newblock Satisfactory graph partition, variants, and generalizations.
\newblock {\em European Journal of Operational Research}, 206(2):271--280,
  2010.

\bibitem{MinApplic4}
Zhigang Cao and Xiaoguang Yang.
\newblock The fashion game: Network extension of matching pennies.
\newblock {\em Theoretical Computer Science}, 540:169--181, 2014.

\bibitem{MajApplic2}
Luca Cardelli and Attila Csik{\'a}sz-Nagy.
\newblock The cell cycle switch computes approximate majority.
\newblock {\em Scientific reports}, 2:656, 2012.

\bibitem{Infection}
Carmen~C Centeno, Mitre~C Dourado, Lucia~Draque Penso, Dieter Rautenbach, and
  Jayme~L Szwarcfiter.
\newblock Irreversible conversion of graphs.
\newblock {\em Theoretical Computer Science}, 412(29):3693--3700, 2011.

\bibitem{MinApplic3}
Jacques Demongeot, Julio Aracena, Florence Thuderoz, Thierry-Pascal Baum, and
  Olivier Cohen.
\newblock Genetic regulation networks: circuits, regulons and attractors.
\newblock {\em Comptes Rendus Biologies}, 326(2):171--188, 2003.

\bibitem{extra2}
Michal Feldman, Nicole Immorlica, Brendan Lucier, and S.~Matthew Weinberg.
\newblock {Reaching Consensus via Non-Bayesian Asynchronous Learning in Social
  Networks}.
\newblock In {\em Approximation, Randomization, and Combinatorial Optimization.
  Algorithms and Techniques (APPROX/RANDOM 2014)}, volume~28 of {\em Leibniz
  International Proceedings in Informatics (LIPIcs)}, pages 192--208, Dagstuhl,
  Germany, 2014. Schloss Dagstuhl--Leibniz-Zentrum f{\"u}r Informatik.

\bibitem{Goles2}
Fran{\c{c}}oise Fogelman, Eric Goles, and G{\'e}rard Weisbuch.
\newblock Transient length in sequential iteration of threshold functions.
\newblock {\em Discrete Applied Mathematics}, 6(1):95--98, 1983.

\bibitem{majority}
Silvio Frischknecht, Barbara Keller, and Roger Wattenhofer.
\newblock Convergence in (social) influence networks.
\newblock In {\em International Symposium on Distributed Computing}, pages
  433--446. Springer, 2013.

\bibitem{majOther2}
Bernd G{\"a}rtner and Ahad~N Zehmakan.
\newblock Color war: Cellular automata with majority-rule.
\newblock In {\em International Conference on Language and Automata Theory and
  Applications}, pages 393--404. Springer, 2017.

\bibitem{MajOther1}
Bernd G{\"a}rtner and Ahad~N Zehmakan.
\newblock Majority model on random regular graphs.
\newblock In {\em Latin American Symposium on Theoretical Informatics}, pages
  572--583. Springer, 2018.

\bibitem{Goles}
Eric Goles and Jorge Olivos.
\newblock Periodic behaviour of generalized threshold functions.
\newblock {\em Discrete Mathematics}, 30(2):187--189, 1980.

\bibitem{MajApplic4}
Mark Granovetter.
\newblock Threshold models of collective behavior.
\newblock {\em American Journal of Sociology}, 83(6):1420--1443, 1978.

\bibitem{hedetniemi}
Sandra~M Hedetniemi, Stephen~T Hedetniemi, KE~Kennedy, and Alice~A Mcrae.
\newblock Self-stabilizing algorithms for unfriendly partitions into two
  disjoint dominating sets.
\newblock {\em Parallel Processing Letters}, 23(01):1350001, 2013.

\bibitem{votingtime}
Dominik Kaaser, Frederik Mallmann-Trenn, and Emanuele Natale.
\newblock On the voting time of the deterministic majority process.
\newblock In {\em 41st International Symposium on Mathematical Foundations of
  Computer Science (MFCS 2016)}, 2016.

\bibitem{majorityW}
Barbara Keller, David Peleg, and Roger Wattenhofer.
\newblock How even tiny influence can have a big impact!
\newblock In {\em International Conference on Fun with Algorithms}, pages
  252--263. Springer, 2014.

\bibitem{MajApplic5}
David Kempe, Jon Kleinberg, and {\'E}va Tardos.
\newblock Maximizing the spread of influence through a social network.
\newblock In {\em Proceedings of the ninth ACM SIGKDD international conference
  on Knowledge discovery and data mining}, pages 137--146. ACM, 2003.

\bibitem{KPRanticoor}
Jeremy Kun, Brian Powers, and Lev Reyzin.
\newblock Anti-coordination games and stable graph colorings.
\newblock In {\em International Symposium on Algorithmic Game Theory}, pages
  122--133. Springer, 2013.

\bibitem{votermodel}
Thomas~M Liggett.
\newblock {\em Stochastic interacting systems: contact, voter and exclusion
  processes}, volume 324.
\newblock Springer Science \& Business Media, 2013.

\bibitem{ising}
Barry~M McCoy and Tai~Tsun Wu.
\newblock {\em The two-dimensional Ising model}.
\newblock Courier Corporation, 2014.

\bibitem{probability}
Michael Mitzenmacher and Eli Upfal.
\newblock {\em Probability and computing: Randomization and probabilistic
  techniques in algorithms and data analysis}.
\newblock Cambridge university press, 2017.

\bibitem{extra3}
Elchanan Mossel, Joe Neeman, and Omer Tamuz.
\newblock Majority dynamics and aggregation of information in social networks.
\newblock {\em Autonomous Agents and Multi-Agent Systems}, 28(3):408--429,
  2014.

\bibitem{useless}
Arpan Mukhopadhyay, Ravi~R Mazumdar, and Rahul Roy.
\newblock Voter and majority dynamics with biased and stubborn agents.
\newblock {\em Journal of Statistical Physics}, 181(4):1239--1265, 2020.

\bibitem{Ahad2018}
Ahad N~Zehmakan.
\newblock Opinion forming in {E}rd{\"o}s-{R}{\'e}nyi random graph and
  expanders.
\newblock In {\em 29th International Symposium on Algorithms and Computations}.
  Schloss Dagstuhl-Leibniz-Zentrum fuer Informatik GmbH, Wadern/Saarbruecken,
  2018.

\bibitem{minority}
P{\'a}l~Andr{\'a}s Papp and Roger Wattenhofer.
\newblock {Stabilization Time in Minority Processes}.
\newblock In {\em 30th International Symposium on Algorithms and Computation
  (ISAAC 2019)}, volume 149 of {\em Leibniz International Proceedings in
  Informatics (LIPIcs)}, pages 43:1--43:19, Dagstuhl, Germany, 2019. Schloss
  Dagstuhl--Leibniz-Zentrum f{\"u}r Informatik.

\bibitem{minorityW}
P{\'a}l~Andr{\'a}s Papp and Roger Wattenhofer.
\newblock {Stabilization Time in Weighted Minority Processes}.
\newblock In {\em 36th International Symposium on Theoretical Aspects of
  Computer Science (STACS 2019)}, volume 126 of {\em Leibniz International
  Proceedings in Informatics (LIPIcs)}, pages 54:1--54:15, Dagstuhl, Germany,
  2019. Schloss Dagstuhl--Leibniz-Zentrum f{\"u}r Informatik.

\bibitem{prop}
P{\'a}l~Andr{\'a}s Papp and Roger Wattenhofer.
\newblock {A General Stabilization Bound for Influence Propagation in Graphs}.
\newblock In {\em 47th International Colloquium on Automata, Languages, and
  Programming (ICALP 2020)}, volume 168 of {\em Leibniz International
  Proceedings in Informatics (LIPIcs)}, pages 90:1--90:15, Dagstuhl, Germany,
  2020. Schloss Dagstuhl--Leibniz-Zentrum f{\"u}r Informatik.

\bibitem{CA1}
Damien Regnault, Nicolas Schabanel, and {\'E}ric Thierry.
\newblock Progresses in the analysis of stochastic 2d cellular automata: A
  study of asynchronous 2d minority.
\newblock In Lud{\v{e}}k Ku{\v{c}}era and Anton{\'i}n Ku{\v{c}}era, editors,
  {\em Mathematical Foundations of Computer Science 2007}, pages 320--332.
  Springer Berlin Heidelberg, 2007.

\bibitem{CA2}
Damien Regnault, Nicolas Schabanel, and {\'E}ric Thierry.
\newblock On the analysis of “simple” 2d stochastic cellular automata.
\newblock In {\em International Conference on Language and Automata Theory and
  Applications}, pages 452--463. Springer, 2008.

\bibitem{CA3}
Jean-Baptiste Rouquier, Damien Regnault, and {\'E}ric Thierry.
\newblock Stochastic minority on graphs.
\newblock {\em Theoretical Computer Science}, 412(30):3947--3963, 2011.

\bibitem{Schoenebeck}
Grant Schoenebeck and Fang-Yi Yu.
\newblock Consensus of interacting particle systems on {E}rd{\"o}s-{R}{\'e}nyi
  graphs.
\newblock In {\em Proceedings of the Twenty-Ninth Annual ACM-SIAM Symposium on
  Discrete Algorithms}, pages 1945--1964. SIAM, 2018.

\bibitem{propDynamos2}
Ahad~N Zehmakan.
\newblock Target set in threshold models.
\newblock {\em Acta Mathematica Universitatis Comenianae}, 88(3), 2019.

\bibitem{propDynamos1}
Ahad~N Zehmakan.
\newblock Tight bounds on the minimum size of a dynamic monopoly.
\newblock In {\em International Conference on Language and Automata Theory and
  Applications}, pages 381--393. Springer, 2019.

\end{thebibliography}

\newpage

\begin{appendices}

\section{Techniques from Probability Theory} \label{App:Z}

In our proofs, we regularly use basic concepts and techniques from probability theory. In particular, our results also apply the following two well-known lemmas \cite{probability}:

\begin{itemize}

\setlength\itemsep{4pt}

 \item \textit{Union Bound:} for any events $A_1$, $A_2$, ..., $A_k$, we have 
\[ \text{Pr}\left(\bigcup_{i=1}^k A_i \right) \: \leq \: \sum_{i=1}^k \, \text{Pr}(A_i). \]
	
 \item \textit{Chernoff Bound:} let $X_1$, $X_2$, ..., $X_k$ be independent Bernoulli random variables with $\text{Pr}(X_i=1) = \frac{1}{2}$ for all $i$. Then for any $\epsilon \in (0,1)$, we have
\[ \text{Pr}\left(\left|\sum_{i=1}^k  X_i \, - \frac{k}{2}\right| \, \geq \, \epsilon \cdot \frac{k}{2} \right) \: \leq \: 2 \cdot e^{- \frac{1}{6} \cdot \epsilon^2 \cdot k} .\]
\end{itemize}

For convenience, we have stated the Chernoff bound for the simplest case of $\text{Pr}(X_i=1) = \frac{1}{2}$, since this is the case for the vast majority of random variables in our analysis. However, we also apply the general version of the Chernoff bound with $\text{Pr}(X_i=1) = p$ on one occasion in the analysis of the Collection Phase, and we also use the bound with a non-constant $\epsilon$ value in the analysis of our majority process construction for basic switching.

Furthermore, we say that an event happens \textit{with high probability} (\textit{w.h.p.}) if it happens with a probability of at least $1-O\left(\frac{1}{n^{c}}\right)$ for some $c>0$. Note that some works use a more relaxed definition of this concept, already accepting any probability of $1-o(1)$ as w.h.p.. Naturally, our results also hold with this more relaxed definition. 

\section{More details on the basic switching constructions} \label{App:A}

\subsection{Minority process construction}

The analysis of the minority construction is rather straightforward. To set $A_0$ and $B_0$ to the appropriate (different) colors, we only require that $A_0$ has an odd degree (to switch $A_0$ to one color) and $|A_0|>|B|$ (to switch $B_0$ to the other); this is satisfied in our graph. Hence, we can begin the sequence of switches in the graph by switching $A_0$ entirely to one color (w.l.o.g. white) and $B_0$ to the other color.

A Chernoff bound then shows that both $A$ and $B$ initially contains at least $(1-\epsilon) \cdot \frac{m}{2} \, \geq \, \frac{m}{3}$ nodes of both colors w.h.p.. This implies that there are $\frac{m}{3}$ white nodes in $A$ that all want to switch to black, and $\frac{m}{3}$ black nodes in $B$ that all want to switch to white. Until these nodes are switched to this preferred color, they all remain switchable regardless of the current color of their neighbors in $P$.

Another Chernoff bound shows that for any small constant $\epsilon$, the initial number of black nodes in the neighborhood of $P$ is also w.h.p. within $[(\frac{1}{2}-\epsilon) \cdot m, \, (\frac{1}{2}+\epsilon) \cdot m]$. This means that by switching at most $\epsilon \cdot m$ of these switchable nodes in either $A$ or $B$, we can ensure that $P$ has exactly one more black neighbors than white. Recall that for convenience, we added an extra neighbor to $P$ in order for $P$ to have an odd degree, too. A choice of a sufficiently small $\epsilon$ ensures that after this, we still have at least $\frac{m}{3} - \epsilon \cdot m > \frac{m}{4}$ switchable nodes in both $A$ and $B$.

We can then execute the alternating sequence in a similar fashion to the original construction in \cite{minority}. We first switch one of the $\frac{m}{4}$ switchable nodes in $A$ to black; then $P$ will have 1 more black neighbors than white, so we can switch the entire group $P$ to white as a result. We then switch one of the $\frac{m}{4}$ switchable nodes in $B$ to white; $P$ now has 1 more white neighbors than black, so we can switch all nodes in $P$ to black. Selecting the switchable nodes from $A$ and $B$ in an alternating fashion, we can create such an alternating sequence of $2 \cdot \frac{m}{4}$ nodes from $A \cup B$, and after each step of this sequence, we can switch all nodes in $P$ again. Since $P$ consists of $m$ nodes, this provides a minority process of at least $\frac{m}{2} \cdot m$ switches. As we have $m=\Theta(n)$, this implies a stabilization time of $\Omega(n^2)$.

\subsection{Majority process construction}

For majority constructions, let us select a constant $c_0$, and introduce the notation $h:=c_0 \cdot \log{m}$. We then take $h$ further distinct groups $A_{1, 1}$, $A_{1, 2}$, $...$, $A_{1, h}$, with each of them consisting of $\frac{m}{h}$ nodes, and connect them each to group $A_0$ (via a complete bipartite connection). For convenience, we assume that $m$ is divisible by $h$, and that $h$ is an odd number. Besides this, we also add $h$ distinct groups $A_{2, 1}$, $A_{2, 2}$, $...$, $A_{2, h}$, with each of these consisting of $\frac{m}{h}$ nodes, too. For each $i \in \{ 1, ..., h\}$, we connect every node in $A_{1,i}$ to every node in $A_{2,i}$. Altogether, the graph consists of $9m+3$ nodes, so we still have $m = \Theta(n)$.

Now let us consider a specific $i \in \{ 1, ..., h\}$. If $h$ is odd, then with a probability of $\frac{1}{2}$, group $A_{1,i}$ contains more white nodes than black nodes initially. This implies that for each node in $A_{2,i}$, the preferred color is white, and thus we can switch each node in $A_{2,i}$ to white (i.e the nodes that were not already white initially). This event happens independently for different $i$ values, since the $A_{1,i}$ are disjoint; hence we can easily show that w.h.p., there exists an $\hat{i} \in \{ 1, ..., h\}$ such that $A_{2,\hat{i}}$ is indeed switchable to white entirely. In particular, the probability that none of the $A_{2,i}$ is switchable to white is $2^{-h}=2^{-c_0 \cdot \log{m}}$, and since $m=\Theta(n)$ implies $\log{m} \geq \frac{1}{2} \cdot \log{n}$ for $n$ large enough, this probability is at most $2^{-2c_0 \cdot \log{n}}$, and thus it is in $O(\frac{1}{n})$ for a sufficiently large choice of $c_0$.

Furthermore, using Lemma \ref{lem:epsbalanced_initially}, one can show that w.h.p. at least $(\frac{1}{2}-\epsilon) \cdot \frac{m}{h}$ of the nodes in $A_{2,\hat{i}}$ were already black initially. This implies that when we turn $A_{2,\hat{i}}$ entirely white, this increases the number of white nodes in $A_{2,\hat{i}}$ by $(\frac{1}{2}-\epsilon) \cdot \frac{m}{h} = \Theta(\frac{m}{\log{m}})$ at least.

Now consider a node $v \in A_{1,\hat{i}}$. Each such node has the same neighborhood: $m+1$ neighbors in $A_0$, and $\frac{m}{h}$ neighbors in $A_{2,\hat{i}}$, giving a total degree of $d_v=m+\frac{m}{h}+1$. Note that we have $m < d_v < 2m$ for a sufficiently large $m$.

As the next step, we show that the neighborhood of $v$ is relatively balanced initially. We need a slightly stronger bound here than in the previous cases, so we now apply the Chernoff bound with a non-constant $\epsilon$ value. 
We can choose, say, $\epsilon:=m^{-2/5}$; then the Chernoff bound shows that the probability of $v$ having more than $(\frac{1}{2} + m^{-2/5}) \cdot d_v$ black neighbors initially is at most
\[ 2 \cdot e^{- \frac{1}{6} \cdot m^{-4/5} \cdot d_v } \, \leq \, 2 \cdot e^{-\frac{1}{6} \cdot m^{-4/5} \cdot m } \, = \,  2 \cdot e^{-\frac{1}{6} \cdot m^{1/5}} \,.\]
Furthermore, note that 
\[ \left(\frac{1}{2} + m^{-2/5}\right) \cdot d_v \, = \, \frac{1}{2} \cdot d_v + m^{-2/5} \cdot d_v \, < \, \frac{1}{2} \cdot d_v + m^{-2/5} \cdot 2m \, = \, \frac{1}{2} \cdot d_v + 2 \cdot m^{3/5} \, , \]
so the same upper bound holds for the probability that the number of black nodes is at least $\frac{1}{2} \cdot d_v + 2m^{\frac{3}{5}}$. Hence we can claim w.h.p. that initially, the number of black nodes in the neighborhood of $A_{1,\hat{i}}$ is larger by at most $2 \cdot m^{3/5}$ than the expected value.

Recall that we have turned the entire $A_{2,\hat{i}}$ white, increasing the number of white nodes in $A_{2,\hat{i}}$ by at least $\Theta(\frac{m}{\log{m}})$. Also, note that $\Theta(\frac{m}{\log{m}}) \, > \, 2 \cdot m^{3/5}$ for $m$ large enough. Therefore, if $A_{1,\hat{i}}$ had at least $\frac{1}{2} \cdot (m+\frac{m}{h}+1) - 2 \cdot m^{3/5}$ white neighbors initially, then after increasing this by $\Theta(\frac{m}{\log{m}})$, the group $A_{1,\hat{i}}$ has more white neighbors that black. This allows us to switch the entire $A_{1,\hat{i}}$ to white, too.

We can then apply a very similar argument on the group $A_0$. Altogether, a node $v \in A_0$ has $d_v=2m$ neighbors, and a Chernoff bound shows that at least $m-2 \cdot m^{3/5}$ of these are already white initially. Lemma \ref{lem:epsbalanced_initially} proves that $A_{1,\hat{i}}$ had at least $(\frac{1}{2}-\epsilon) \cdot \frac{m}{h} = \Theta(\frac{m}{\log{m}})$ black nodes initially, so when turning $A_{1,\hat{i}}$ entirely to white, we increase the number of white nodes in $A_{1,\hat{i}}$ by at least $\Theta(\frac{m}{\log{m}})$. This results in at least $m-2 \cdot m^{3/5}+\Theta(\frac{m}{\log{m}}) \, > \, m =\frac{1}{2} \cdot d_v$ white neighbors for $A_0$, so we can switch each node in $A_0$ white.

From here, our construction follows the same idea as the minority case. Turning $A_0$ white already ensures that every black node in $A$ is switchable to white. In a symmetric manner, we can turn each node in $B_0$ black, ensuring that every white node in $B$ is switchable to black. Then we can use the same alternating method as in the minority construction, which implies that we can switch the group $P$ a total of $\Theta(m)$ times altogether. Since we still have $m = \Theta(n)$, this again provides a sequence of $\Omega(m^2)=\Omega(n^2)$ switches.

\section{More details on the proportional switching construction} \label{App:B}

\subsection{Overall analysis}

Let us first discuss the number of nodes in our construction.

Recall that in the Opening Phase, we obtain our $S_0$ by taking $m_0$ independent copies of the tree described in Section \ref{sec:opening}. With $p$, $i$ and $2^{i+1}-1$ being constants, the whole phase requires only $O(m_0)$ nodes.

The Collection Phase then creates a set $S_1$ on $m_1:=\frac{m_0}{4 \cdot d_0}$ nodes; this already determines that $|S_{2,1}|=m_1$, too. Each level of the Growing Phase has the same size, i.e. $|S_{2,i}|=m_1$ for every $i \in \{1, ..., \ell \}$. To reach an output degree of, say, $\frac{1}{2} \cdot m_1$ for every node in $S_{2,\ell}$, we need about $\ell \approx \log{m_1}$ distinct levels.

Then in the Control Phase, we create a set on $|S_3|=m_3=\frac{m_1}{\alpha}$ nodes. Finally, the Simulation Phase uses a \textsc{prop} construction on $m=\frac{1}{2} \cdot m_3$ nodes.

This implies that $m_1=2 \alpha \cdot m$ for the size of the levels $S_{2,i}$. Since $\alpha$ is a constant, this results in a Growing Phase construction of $O(\log{m_1})=O(\log{m})$ distinct levels of size $m_1$, which is altogether still only $O(m \cdot \log{m})$ nodes. Finally, the Opening Phase adds another $O(m_1 \cdot \log{n})$ nodes to this; if $m = \Omega(\sqrt{n})$ and thus $\log{n} \leq 2 \log{m}$, then this is still only $O(m \cdot \log{m})$ nodes. Note that some of the phases also require two distinct copies of the previous parts of the graph, but even with this, each phase only appears constantly many times in our construction. Hence the total number of nodes in the graph is $O(m \cdot \log{m})$, which allows for a choice of $m:=\Omega(\frac{n}{\log{n}})$ with the appropriate constants.

Also, note that there are only constantly many distinct points of the construction where we point out that an event happens w.h.p.. In particular, we use one such assumption in the Collection Phase when we discuss the number of black neighbors developed in the Opening Phase, another one in the Growing Phase when we assume that all output neighborhoods are initially $\epsilon$-balanced, and a final one in the Control Phase when we assume that for each $v\in S_3$, the set of neighbors in $S_{2,\ell}^w$ is initially $\epsilon$-balanced. Our final construction only contains constantly many copies of each of these phases. Thus we only make constantly many such assumptions altogether, which means that we can simply use a union bound to show that w.h.p. all of these assumptions will hold simultaneously. Therefore, we can indeed claim that our entire construction will w.h.p. behave as discussed.

\subsection{Majority constructions to minority constructions}

While our proportional lower bound construction was presented for majority processes, we can easily adapt it to the case of minority processes. Note that each of the first $4$ phases in our construction is a bipartite graph, so we can simply take one of the two color classes in the construction, and swap the role of the two colors in this color class to obtain the same behavior. This technique can be demonstrated most easily in the Growing Phase: if we can make each node in $S_{2,1}$ black, then this allows us to switch each node in $S_{2,2}$ white, then each node in $S_{2,3}$ black again, each node in $S_{2,4}$ white again, and so on. In the end, we can obtain a set $S_3$ with the same property as before.

The original \textsc{prop} construction from \cite{prop} is also a bipartite graph, and from a different initial ordering, it also provides an example sequence where stabilization lasts for $n^{1+f(\lambda)-\varepsilon}$ steps for minority processes. Hence, in an identical way to majority processes, we can now use our set $S_3$ in the Control Phase in order to force the \textsc{prop} construction to first take the desired initial colors, and then we can execute this sequence of switches. This provides an example construction to show the same lower bound in case of minority processes.

One can also observe that the graph presented in Section \ref{sec:simple_maj} (i.e. the lower bound construction for majority processes with basic switching) is also a bipartite graph, and thus a similar method also allows us to convert this to a minority construction that shows a stabilization time of $\Theta(n^2)$. As such, the construction of Section \ref{sec:simple_min} is in fact not needed for the completeness of the paper, and could instead be replaced by a slight modification of the construction in Section \ref{sec:simple_maj}. Nonetheless, we decided to still include the Section \ref{sec:simple_min} construction in the paper because it provides a notably simpler proof of the lower bound in case of minority processes.

\subsection{Details of the Opening Phase}

The main idea of the Opening Phase has already been discussed in Section \ref{sec:lower}. Each node $v \in S_0$ is obtained as the root of a balanced binary tree. By taking all nodes in a leaf-to-root fashion in this tree and turning them black whenever possible, we ensure that the probability of turning a specific node black after $i$ layers is described by the recurrence $p_{i+1} = \frac{1}{2}+\frac{1}{2} \cdot p_i^{\ 2}$. For any desired $p<1$, a constant number of layers is sufficient to ensure that the root $v$ becomes black with a probability of $p_i > p$ in the end.

Thus our construction of $S_0$ consists of $m_0$ independent trees of $i$ layers, where each node in the tree has 2 new neighbors in the following layer (except for the last layer). The set $S_0$ consists of the root nodes of each of these $m_0$ distinct trees. With both $p$ and $i$ being constants, the phase only requires $O(m_0)$ nodes altogether.

Note that it is not straightforward to generalize this technique for $\lambda$ values higher than $\frac{1}{3}$. E.g. for any $\lambda<\frac{1}{2}$, one could devise a similar construction where each node has $3$ input neighbors $u_1$, $u_2$, $u_3$ (since $\lambda<\frac{1}{2}$ implies $\frac{1+\lambda}{2} \leq \frac{3}{4}$), and we similarly end up with a tree of nodes with degree $4$. However, this provides the recurrence $p_{i+1} = \frac{1}{2}+\frac{1}{2} \cdot p_i^{\ 3}$ for the values $p_i$, which does not converge to $1$, but instead to a limit of $\frac{\sqrt{5}-1}{2}$. Hence, this technique does not allow us to turn each node in $S_0$ black with an arbitrarily high probability $p$.

\subsection{Details of the Collection Phase}

Overall, the Collection Phase is the simplest phase in our construction. The set $S_1$ is simply a set of $m_1$ nodes, each having a degree of $5 \cdot d_0$. An Opening Phase of size $m_0=m_1 \cdot 4 \cdot d_0$ provides enough nodes such that each $v \in S_1$ can be connected to $4 \cdot d_0$ distinct nodes in $S_0$. Besides this, each $v$ will also have $d_0$ output edges to the next phase.

If each neighbor of $v$ in $S_0$ becomes black with a probability of at least $p$, then $v$ has at least $p \cdot 4 \cdot d_0$ black neighbors in $S_0$ in expectation. We can then use a Chernoff bound to show that the probability is heavily concentrated around this expectation. Note that this requires the Chernoff bound on general Bernoulli random variables $X_i$ with $\text{Pr}(X_i=1)=p$; for simplicity, in Appendix \ref{App:Z}, we have only stated the bound for the simplest case of $p=\frac{1}{2}$.

Let us select $p=\frac{15}{16}$ in our Opening Phase. Let $\epsilon < \frac{5}{48}$ in order to ensure $\frac{4}{5}\epsilon < \frac{3}{4}-\frac{2}{3}$, and let us define $\hat{\epsilon}:=\frac{\epsilon}{p}$. Furthermore, let $X$ denote the number of black neighbors in $S_0$. Then the Chernoff bound shows that the probability of differing by more than an $\hat{\epsilon}$ multiplicative factor from the expected value is
\[  \text{Pr}\left( X \leq (1-\hat{\epsilon}) \cdot p \cdot 4 \cdot d_0 \right) \: \leq \: e^{-\frac{\hat{\epsilon}^2 \cdot p \cdot 4 \cdot d_0}{2}} = e^{-2\hat{\epsilon}^2 \cdot p \cdot c_0 \cdot \log{n}} = n^{-2\hat{\epsilon}^2 \cdot p \cdot c_0} \,. \]
We can easily ensure that this is in $O(n^{-2})$ by choosing $c_0$ high enough such that $\hat{\epsilon}^2 \cdot p \cdot c_0 \geq 1$. Note that $(1-\hat{\epsilon}) \cdot p \cdot 4 \cdot d_0 = (p-\epsilon) \cdot 4 \cdot d_0$ due to the definition of $\hat{\epsilon}$.

With $d_v=5 \cdot d_0$, this implies a ratio of $\frac{(p-\epsilon) \cdot 4 \cdot d_0}{5 \cdot d_0}=\frac{3}{4}-\frac{4}{5}\epsilon \, > \, \frac{2}{3}$ blacks in the neighborhood, so the event that we cannot switch $v$ black only has a probability of $O(n^{-2})$. Taking a union bound over all $v \in S_1$, we get that we can switch the entire $S_1$ black with a probability of $1-O(n^{-1})$.

Note that we can also easily generalize this phase for any $\lambda<\frac{1}{2}$. A value of $\lambda<\frac{1}{2}$ still implies $\frac{1+\lambda}{2} < \frac{3}{4}$, so we only need to ensure $(p-\epsilon)\cdot \frac{4}{5} \, > \, \frac{3}{4}$ in this case. This is achieved by any $p>\frac{15}{16}$ and a sufficiently small $\epsilon$.

\subsection{Details of the Growing Phase}

The Collection Phase already gives us a set $S_1$ on $m_1$ nodes with each $v \in S_1$ having $d_0 = c_0 \cdot \log{n}$ output edges. We now describe the Growing Phase in a more general form than in Section \ref{sec:lower} to address the case of an arbitrary $\lambda$ value with $\lambda<\frac{1}{2}$. As the key idea of the phase, we select a small parameter $\mu>0$, and we design the levels such that the output degree in $S_{2,i+1}$ is always a $(1+\mu)$ factor larger than the output degree in $S_{2,i}$. Note that in Section \ref{sec:lower}, we discussed the special case of $\mu=1$.

We then build the level-based construction described in Section \ref{sec:lower}. We first select $S_{2,1}=S_1$. We then connect $S_{2,1}$ and $S_{2,2}$ as a $d_0$-regular bipartite graph, we connect $S_{2,2}$ and $S_{2,3}$ as a $(1 + \mu) \cdot d_0$-regular bipartite graph, we connect $S_{2,3}$ and $S_{2,4}$ as a $(1 + \mu)^2 \cdot d_0$-regular bipartite graph, and so on; $S_i$ and $S_{i+1}$ forms a $(1 + \mu)^{i-1} \cdot d_0$-regular bipartite graph. We can always select an arbitrary one among the different possible bipartite graphs to implement the connection between the given levels.

After at most $\log_{(1+\mu)}{m_1}$ such levels, we reach a level $S_{2,\ell}$ where the degree of each node is at least $\frac{1}{2} \cdot m_1$; we will use this last level for the next phase of our construction. Note that since $m_1=\Omega(\frac{n}{\log{n}})$, we also know that $\ell=O(\log{n})$. As each of our levels consist of the same number of nodes $m_1$, we only require $O(m_1 \cdot \log{m_1})$ nodes for this phase altogether. With our choice of $m_0=\Theta(n)$ and $m_1=\Theta(\frac{m_0}{\log{n}})$, we have $m_1=\Theta(\frac{n}{\log{n}})$, and thus $O(m_1 \cdot \log{m_1})$ is indeed smaller than $n$ for the appropriate choice of constants.

To show that w.h.p. we can turn each node black in every level $S_{2,i}$, we use an induction. Initially, we already know that w.h.p. we can turn each node in $S_1$ black. Furthermore, we will assume that the outputs of each node in every level are initially $\epsilon$-balanced. Note that since each node in this phase already has at least $c_0 \cdot \log{n}$ output neighbors, and there are at most $n$ nodes altogether, we can apply Lemma \ref{lem:epsbalanced_initially} to show that w.h.p. this claim holds in our graph.

Now let us consider a general level of the construction. Recall that for a general node $v$, we use $d$ to denote the degree to the previous level, which means that $v$ has $(1+\mu) \cdot d$ output edges and a total degree of $d_v=(2+\mu) \cdot d$. If the outputs are $\epsilon$-balanced initially, then at least $(\frac{1}{2}-\epsilon) \cdot (1 + \mu) \cdot d$ out of the $(1 + \mu) \cdot d$ outputs are already black initially. Our induction hypothesis states that we can turn all the $d$ previous-level neighbors of $v$ black. This altogether amounts to at least $(1 + (\frac{1}{2}-\epsilon) \cdot (1 + \mu)) \cdot d$ black neighbors. Thus to show that $v$ is switchable to black at this point, we need 
\[  \frac{(1 + (\frac{1}{2}-\epsilon) \cdot (1 + \mu)) \cdot d}{(2 + \mu) \cdot d} \; \geq \; \frac{1+\lambda}{2}\,.\]
After expansion and simplification, this gives $2 \cdot \lambda + 2 \cdot \epsilon \cdot (1-\mu) + \mu \cdot \lambda \, \leq \, 1$. For any value of $\lambda<\frac{1}{2}$, we can ensure this with a sufficiently small choice of $\mu$ and $\epsilon$. Hence after $S_{2,i-1}$ becomes black, we can also turn $S_{2,i}$ entirely black.

We point out that this growing phase construction does not require a new probabilistic statement with each new level: we only use the fact that $S_1$ can be switched entirely black w.h.p., and that the output neighborhood of each node is $\epsilon$-balanced initially (which follows from Lemma \ref{lem:epsbalanced_initially}). From this, the rest of our claims follow deterministically.

\subsection{Details of the Control Phase}

Intuitively, the base idea of the Control Phase is to make the output edges such an insignificant part of the neighborhood of $S_3$ that the switchability of the nodes $S_3$ is always controlled solely by the connections to $S_{2,\ell}^b$ and $S_{2,\ell}^w$. Since we have $\frac{1+\lambda}{2} = \frac{2}{3} -\delta$ for some constant $\delta>0$, we can achieve this by ensuring that the current conflicts to $S_{2,\ell}^b$ and $S_{2,\ell}^w$ always amount to almost $\frac{2}{3}$ of the total degree.

This phase already requires us to create two different copies of the previous $3$ phases. That is, besides the instance of the first three phases that allows us to switch each node in $S_{2,\ell}^b$ black, we also create another Opening, Collection and Growing Phase for the color white in a symmetric manner, which in the end allows us to switch all the nodes in the final set $S_{2,\ell}^w$ white. This only doubles the total number of nodes that we use for the first 3 phases, and thus it does not affect the magnitude of the final size of our construction.

We choose the size of these three-phase constructions such that the size of $S_{2,\ell}^b$ is $\alpha \cdot m$, while the size of $S_{2,\ell}^w$ is $2 \alpha \cdot m$. For simplicity, we choose $m_1$ to denote the size of the larger of the two sets, i.e. $2 \alpha \cdot m$. For the other copy of the first three phases (i.e the one ending with $S_{2,\ell}^b$), we in fact only require half as many nodes, i.e. levels of size $\frac{m_1}{2}$ in the Growing Phase.

Note that $S_{2,\ell}^b$ is the last level of a Growing Phase on $\alpha \cdot m$ nodes, so each node in $S_{2,\ell}^b$ has an output degree of at least $\frac{\alpha}{2} \cdot m$. Since we have $|S_3|=2m$, for a sufficiently large $\alpha$ (i.e. $\alpha \geq 4$), it is indeed possible to connect each node in $S_3$ to every node in $S_{2,\ell}^b$, as the nodes in $S_{2,\ell}^b$ do have a sufficiently large output degree for this. Thus we can indeed ensure that each node in $S_3$ has $\alpha \cdot m$ edges to $S_{2,\ell}^b$.

Similarly, $S_{2,\ell}^w$ is the last level of a Growing Phase on $2\alpha \cdot m$ nodes, hence each node in $S_{2,\ell}^w$ has an output degree of $\alpha \cdot m$ at least. Again, this output degree shows that a choice of $\alpha \geq 2$ allows us to connect each node in $S_{2,\ell}^w$ to all the $2m$ nodes in $S_3$.

Furthermore, note that we assume that for each node $v \in S_3$, the set of neighbors of $v$ in $S_{2,\ell}^w$ is initially $\epsilon$-balanced. Since $v$ has $2\alpha \cdot m$ neighbors in $S_{2,\ell}^w$ which is significantly larger than $\Theta(\log{n})$, we can easily make such an assumption; Lemma \ref{lem:epsbalanced_initially} shows that w.h.p. it holds for all nodes $v \in S_3$.

Also, we point out that this is a phase that we cannot generalize to larger $\lambda$ values up to $\frac{1}{2}$: the fraction $\frac{2\alpha \cdot (1-\epsilon)}{3\alpha +1}$ is upper-bounded by $\frac{2}{3}$, and any other configuration of connections to $S_{2,\ell}^b$ and $S_{2,\ell}^w$ would either not make $S_3$ switchable to black in the first place, or it would not be enough to switch it back to white later.

\subsection{Details of the Simulation Phase}

Most aspects of the Simulation Phase have already been discussed in Section \ref{sec:lower}. For each node $v$ in the \textsc{prop} construction, we add $\frac{1}{2} \cdot \frac{1+\lambda}{1-\lambda} \cdot d_v'$ new neighbors in both $S_3^{b}$ and $S_3^{w}$, which first allows us to force $v$ to take the desired initial color, and then to make the new part of the neighborhood balanced. This allows us to run the \textsc{prop} construction for $\lambda'=\frac{2 \lambda}{1-\lambda}$, providing a sequence of $m^{1+f( \frac{2 \lambda}{1-\lambda} ) -\varepsilon}$ steps for any $\varepsilon>0$. Since we have $m=\Omega(\frac{n}{\log{n}})$ and we can get rid of the logarithmic factor with a smaller choice of $\varepsilon$, this shows a lower bound of $n^{1+f( \frac{2 \lambda}{1-\lambda} ) -\varepsilon}$.

Recall that we have only discussed the Simulation Phase for the \textsc{prop} construction nodes that are initially black in the black box construction. In practice, we also need an entirely separate copy of the first 4 phases in order to set the initial color of the remaining \textsc{prop} nodes white. That is, we create another instance of the first 4 phases in a symmetric manner, similarly to the doubling step of the Control Phase. This now allows us to turn all levels in the Growing Phase white, and then obtain a set $S'_3$ where first every node can be switched to white, and then half of the nodes can be switched back to black. This again only doubles the number of nodes required in the first 4 phases, which does not affect the magnitude of the size of the graph.

One might also wonder if we can generalize this phase to larger $\lambda$ values by connecting our \textsc{prop} nodes to a fewer number of nodes in $S_3$, and instead using the fact that the nodes have an initially $\epsilon$-balanced neighborhood within the \textsc{prop} construction. However, our processes are sequential, and thus we could only apply this argument on the first nodes that are switched to their preferred initial color in the \textsc{prop} graph. The later nodes, on the other hand, will have a severely biased neighborhood due to the fact that we have already set many of their neighbors in the \textsc{prop} construction to the desired initial color.

\section{Brief discussion of the function $f(\lambda)$} \label{App:C}

For the sake of completeness, we also describe the function $f(\lambda)$ that was introduced in \cite{prop} and used in Theorem \ref{th:lower}.

The domain of $f$ is the open interval $\lambda \in (0,1)$, and the image of $f$ is also $(0,1)$. On $(0,1)$ the function $f$ is continuous, monotonously decreasing and convex, with $\lim_{\lambda \rightarrow 0}f(\lambda)=1$ and $\lim_{\lambda \rightarrow 1}f(\lambda)=0$ in the limits. As such, the bounds of $n^{1+f(\lambda) \pm \varepsilon}$ in \cite{prop} describe a transition from quadratic to linear behavior as $\lambda$ goes from $0$ to $1$.

The concrete formula of the function is given in terms of a parameter $\varphi$ such that $\varphi \in (0,\frac{1-\lambda}{2}]$. That is, the authors describe the stabilization time as a function of $\varphi$, and they show that stabilization time is maximal when the optimum $\varphi$ is chosen. In particular, the function $f$ is defined as
\[ f(\lambda) := \max_{\varphi \in (0,\frac{1-\lambda}{2}]} \; \frac{\log\left( \frac{1-\varphi}{\lambda+\varphi} \right)}{\log\left( \frac{1-\varphi}{\varphi} \right)} \, . \]
A derivative of this expression leads to an equation that cannot be solved with elementary methods, and thus there is no straightforward way to present $f(\lambda)$ in a simple closed form. The plot of $f(\lambda)$ is illustrated in Figure \ref{fig:f}.

Recall that in our lower bound presented in Section \ref{sec:lower}, we first apply the transformation $\lambda \rightarrow \frac{2 \lambda}{1 - \lambda}$, which maps the interval $(0,1)$ into $(0, \frac{1}{3})$; we only call the function $f$ after this transformation. The resulting function $f(\frac{2 \lambda}{1 - \lambda})$, as visible on the left side of Figure \ref{fig:func}, is a continuous, monotonously decreasing, convex function on the domain $(0, \frac{1}{3})$. The image of the function is the entire $(0,1)$, since we now have $\lim_{\lambda \rightarrow 0}f(\frac{2 \lambda}{1 - \lambda})=1$ and $\lim_{\lambda \rightarrow \frac{1}{3}}f(\frac{2 \lambda}{1 - \lambda})=0$ in the limits. As such, our lower bound exhibits a similar transition from quadratic to linear behavior on the interval $(0, \frac{1}{3})$.

\begin{figure}
\centering
	\includegraphics[width=0.63\textwidth]{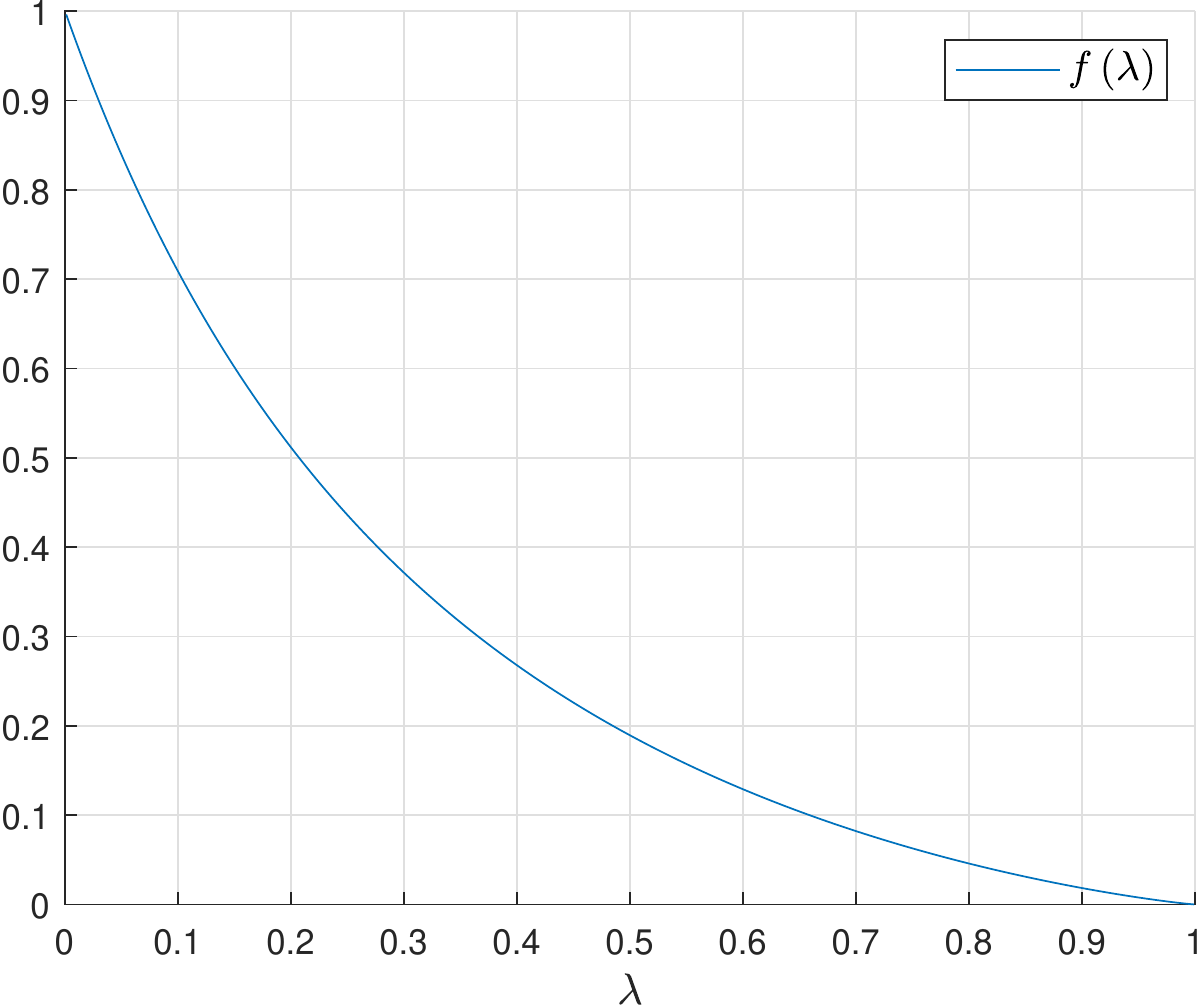}
	\caption{Illustration of the function $f(\lambda)$ introduced in \protect \cite{prop}.}
	\label{fig:f}
\end{figure}

\end{appendices}

\end{document}